\documentclass[10pt]{article}

\usepackage{lscape, youngtab, ulem}

\usepackage{epsfig,cancel,amsthm,amssymb}
\usepackage{color,soul,bm}
\usepackage{wrapfig}
\usepackage{graphicx,framed}
\usepackage[colorlinks=true, pdfstartview=FitV, linkcolor=black, citecolor=black, urlcolor=blue]{hyperref}
\definecolor{shadecolor}{rgb}{0.95, 0.95, 0.86}
\definecolor{darkgreen}{rgb}{0.2, 0.5,  0}

\def\&{\vspace{-5pt}&}

\makeatletter
\@addtoreset{equation}{section}
\makeatother

\def \eqref#1{(\ref{#1})}
\def \& {&\hspace{-10pt}}
\def \Ai{ {\mathrm {Ai}}}

\def \wt{\widetilde}

\newcommand{\lla}{\langle\langle} 
\newcommand{\rra}{\rangle\rangle} 
\newcommand{\bt}{\beta}

\newcommand{\pa}{\partial}
\newcommand{\p}{\partial}
\newcommand{\bdt}{{\bf t}}
\newcommand{\bds}{{\bf {s}}}

\newcommand{\bdzero}{{\bf 0}}
\newcommand{\fs}{{s}}

\newcommand{\br}{{\mathbb R}}

\newcommand{\nn}{\nonumber}
\newcommand{\Rsol}{\mathcal R}
\newtheorem{theorem}{Theorem}[section]
\newtheorem{example}[theorem]{Example}

\newtheorem{exercise}[theorem]{Exercise}

\newtheorem{lemma}[theorem]{Lemma}
\newtheorem{remark}[theorem]{Remark}

\newtheorem{proposition}[theorem]{Proposition} 
\newtheorem{corollary}[theorem]{Corollary} 
\newtheorem{definition}[theorem]{Definition}

\def\le{\left}
\def\ri{\right}
\def\ds{\displaystyle}

\def\res{\mathop{\mathrm {res}}\limits_}

\def\bt{\begin{theorem}}
\def\et{\end{theorem}}

\def\bc{\begin{corollary}}
\def\ec{\end{corollary}}

\def\bx{\begin{example}\small}
\def\ex{\end{example}}

\def\bxr{\begin{exercise}\small}
\def\exr{\end{exercise}}
\def\bl{\begin{lemma}}
\def\el{\end{lemma}}
\def\bd{\begin{definition}}
\def\ed{\end{definition}}
\def\bp{\begin{proposition}}
\def\ep{\end{proposition}}

\def\br{\begin{remark}}
\def\er{\end{remark}}

\def\be{\begin{eqnarray}}
\def\ee{\end{eqnarray}}
\def\ov {\overline}

\def\&{\hspace{-15pt}&}
\def\bea{\begin{eqnarray}}
\def\eea{\end{eqnarray}}
\def\beas{\begin{eqnarray*}}
\def\eeas{\end{eqnarray*}}

\def \pa{\partial}

\def\l{\lambda  }

\def\1{{\bf 1}}

\def\QED {$\Box$\par\vskip 3pt}

\begin{document}
\title{Correlation functions of the KdV hierarchy and applications to intersection numbers over $\overline{\mathcal M}_{g,n}$}
\author{Marco Bertola, Boris Dubrovin, Di Yang}
\date{}
\maketitle
\begin{abstract}
We derive an explicit generating function of correlation functions of an arbitrary tau-function of the KdV hierarchy. In particular applications, our formulation gives closed formul\ae\ of a new type for the  generating series of intersection numbers of $\psi$-classes as well as of mixed $\psi$- and $\kappa$-classes in full genera.
\end{abstract}

{\small \noindent \textbf{Keywords.} KdV hierarchy; tau-function; wave function; correlation function; intersection numbers.}
\tableofcontents

\section{Introduction and results}
The famed Korteweg--de Vries (KdV) equation
\be\label{KdV-equation}
u_t=u u_x+\frac{1}{12}u_{xxx},
\ee
has long been known to be {\it integrable}. It belongs to an infinite family of pairwise commuting nonlinear evolutionary PDEs called {\it the KdV hierarchy}. The hierarchy can be  described in terms of isospectral deformations of the Lax operator
\be\label{Lax-op}
L= \p_x^2 + 2\,u(x).
\ee
Namely, the $k$-th equation of the KdV hierarchy, in the normalization of the present paper, reads
\begin{eqnarray}
&&
L_{t_{k}}=[A_k, L], \label{GD}\\
&& A_k:=
\frac{1}{(2k+1)!!}\le(L^{\frac{2k+1}{2}}\ri)_+,\quad k\geq 0.\label{A_k}
\end{eqnarray}
Here, the independent variables $t_0,t_1,t_2,...$ are called {\it  times}. The symbol $\le(L^{\frac{2k+1}{2}}\ri)_+$ stands for the differential part of the pseudo-differential operator $L^{\frac{2k+1}{2}}$, see e.g. the book \cite{Dickii} for details. The $k=1$ equation of \eqref{GD} coincides with \eqref{KdV-equation}. As customary in the literature, we shall identify  $t_0$  with the spatial variable $x$.
 
The notion of {\it   tau-function} for the KdV hierarchy was introduced by the Kyoto school \cite{Sato, Hirota, DJMK} during  the 1970s--1980s. 
In 1991, E.\,Witten, in his study of two-dimensional quantum gravity \cite{Witten}, conjectured that the generating function of the intersection numbers of $\psi$-classes on the Deligne--Mumford moduli spaces $\overline{\mathcal M}_{g,n}$ of stable algebraic curves is a tau-function of the KdV hierarchy. Witten's conjecture was later proved by M.\,Kontsevich \cite{Kontsevich}; see \cite{KL,OP, M} for several alternative proofs.  Moreover the so-called ``tau structures" of KdV-like hierarchies became one of the central subjects in the study of the deep relation between integrable hierarchies and Gromov--Witten invariants \cite{DZ-norm,Du2,DLYZ}.

The original motivations of the Witten's conjecture identify the tau-function $\tau=\tau(t_0, t_1, t_2, \dots)$ of a particular solution to the KdV hierarchy  with the partition function of 2D quantum gravity. The time variables $t_0$, $t_1$, $t_2$, \dots ~are identified with coupling constants associated with observables $\tau_0=1$, $\tau_1$, $\tau_2$, \dots ~of the quantum theory. Thus, the quantum correlators of 2D quantum gravity are just logarithmic derivatives of the tau-function
\be\label{corel}
\langle \tau_{k_1}\tau_{k_2}\dots \tau_{k_n}\rangle= \left. \frac{\partial^n \log\tau (t_0, t_1, t_2, \dots)}{\partial t_{k_1}\partial t_{k_2}\dots \partial t_{k_n}}\right|_{t_0=t_1=t_2=\dots=0}.
\ee

Tau-functions of some other solutions to the KdV hierarchy along with the corresponding correlators of the form \eqref{corel} proved to be of interest for other applications, in particular to the study of topology of Deligne--Mumford moduli spaces (see below). The main goal of this paper is to provide a simple algorithm, in the framework of the theory of Lax operator and its eigenfunctions, for computation of $n$--point correlators of an arbitrary solution to the KdV hierarchy.

Let us begin with basics of the theory of KdV tau-functions in the version of M.~Sato {\it et al.} Denote by $\bdt=(t_0,t_1,t_2,\dots)$ the infinite vector of time variables. Let $\mathcal{B}=\mathbb{C}[[t_{k},\,k=0,1,2,\dots]]$ be the Bosonic Fock space. A {\it  Sato tau-function} $\tau(\bdt)$ of the KdV hierarchy is an element in $\mathcal{B}$ satisfying the Hirota bilinear identities
\be\label{bilinear-KdV}
\res{z=\infty} \, \tau(\bdt-[z^{-1}])\,\tau(\wt\bdt+[z^{-1}])
\exp \le({\sum_{j\geq 0} \frac{t_j-\wt t_j}{(2j+1)!!} z^{2j+1}}\ri)
 z^{2p}\, d z=0,\qquad \forall\,\bdt,\wt\bdt,\, p=0,1,2,\dots.
\ee
Here $\bdt-[z^{-1}]:=\le(t_0-z^{-1}, \dots, t_{k}-\frac{(2k-1)!!}{z^{2k+1}},\dots\ri).$ The residue is understood in a formal way, namely, as (minus) the coefficient of  $z^{-1}$ in the formal expansion at $z=\infty$.
Given an arbitrary tau-function $\tau(\bdt)$, then $u(\bdt)= \p_x^2 \log \tau(\bdt)$  is a solution of the KdV hierarchy \eqref{GD}.

Conversely, let $u(\bdt)$ be an arbitrary (formal) solution of the KdV hierarchy \eqref{GD}; then there exists \cite{Dickii} a tau-function $\tau(\bdt)$  such that 
$
 \p_x^2 \log \tau(\bdt)=u(\bdt).
$
The tau-function of $u(\bdt)$ is uniquely determined up to a  {\it gauge freedom}
\be\label{gauge-g}
\tau(\bdt)\mapsto \exp\bigg(\alpha_{-1}+\sum_{j\geq 0} \alpha_{j}\, t_{j}\bigg)\tau(\bdt),
\ee
where {the coefficients $\alpha_j,\,j\geq -1$ are arbitrary constants.} 

Let  $\tau(\bdt)$ be any tau-function of $u(\bdt)$. Define the {\it wave} and {\it dual wave functions} by
\be\label{Sato-tau}
\psi(z;\bdt)=\frac{\tau(\bdt-[z^{-1}])}{\tau(\bdt)} {\rm e}^{\vartheta(z;\bdt)},\quad \psi^*(z;\bdt)=\psi(-z;{\bf t})=\frac{\tau(\bdt+[z^{-1}])}{\tau(\bdt)} {\rm e}^{-\vartheta(z;\bdt)},
\ee
where the phase $\vartheta$ is given by
\be \label{wave-phase-theta}\vartheta(z;\bdt):=\sum_{j=0}^\infty t_j \frac{z^{2j+1}}{(2j+1)!!}.\ee
The gauge freedom \eqref{gauge-g} affects $\psi(z;\bdt)$ by a multiplicative factor of the form
\be
\label{gauge}
g(z)=\exp\le(-\sum_{k=0}^\infty \frac{\alpha_k \, (2k-1)!!}{z^{2k+1}}\ri)=1+\mathcal{O}(z^{-1}),\qquad z\rightarrow \infty.
\ee
The wave functions are (formal) eigenfunctions of the Lax operator
\be
 L \, \psi =z^2 \, \psi, \quad L\, \psi^*=z^2 \,\psi^*. \label{wave-L}
\ee
Their dependence on time variables is specified by 
the following compatible system
\be
\psi_{t_{k}} =A_k \psi, \quad \psi^*_{t_{k}}=-A_k \psi^*, \quad k\geq 0, \label{wave-t}
\ee
where the differential operators $A_k$ are defined in \eqref{A_k}. 
As $z\to \infty $ their (formal) asymptotic behaviours are 
\be
\label{wave-norm} \psi(z;\bdt)=\le( 1+\mathcal{O}(z^{-1}) \ri)  {\rm e}^{\vartheta(z;\bdt)},\quad \psi^*(z;\bdt)=\le( 1+\mathcal{O}(z^{-1}) \ri)  {\rm e}^{-\vartheta(z;\bdt)}.
\ee
Also they satisfy an infinite system of bilinear relations
\bea\label{bilinear-psis}
&& \res{z=\infty} \, \psi(z;\bdt)\, \psi^*(z;\wt \bdt)\,  z^{2p}\, d z=0,\qquad \forall\,\bdt,\,\wt\bdt,\, p=0,1,\dots.
\eea
See e.g.~\cite{Dickii}, Thm.~6.3.8 for proofs of these statements. 
Depending on the context,  the wave and dual wave functions (as well as the $\tau$-function)  can be defined analytically following the approach of G.~Segal and G.~Wilson~\cite{SW}; also see our Remark~\ref{ana-rmk} below. 

\bd For any tau-function $\tau(\bdt)$ of the KdV hierarchy, we call the functions 
\be\label{def-cor-n}
\langle\langle\tau_{k_1}\tau_{k_2}\dots\tau_{k_n}\rangle\rangle (\bdt):=
\frac{\p^n \log \tau}{\p t_{k_1}\dots \p t_{k_n}}(\bdt),\quad  n\geq 1
\ee
and the numbers
\be\label{def-cor-n0}
\langle\tau_{k_1}\tau_{k_2}\dots\tau_{k_n}\rangle:=\langle\langle\tau_{k_1}\tau_{k_2}\dots\tau_{k_n}\rangle\rangle (\bdt=0)
\ee
the {\bf $n$-point correlation functions} and the {\bf $n$-point correlators} of the tau-function respectively. For every $n\geq 1$ we define
{the generating function} of $n$-point correlation functions by
\be
F_n(z_1,\dots,z_n;\bdt):=\sum_{k_1,\dots,k_n=0}^\infty  \langle\langle \tau_{k_1} \dots \tau_{k_n}\rangle\rangle (\bdt)\, \frac{(2k_1+1)!!}{z_1^{2k_1+2}}\cdots \frac{(2k_n+1)!!}{z_n^{2k_n+2}}.
\ee
Evaluating it at $\bdt=0$
\be
F_n(z_1,\dots,z_n):=F_n(z_1,\dots,z_n;\bdzero)
\ee
one obtains generating functions of $n$-point correlators of $\tau(\bdt)$.
\ed
\noindent The notation \eqref{def-cor-n} is borrowed from the literature in quantum gravity \cite{Witten}.

We are now in a position to formulate the first main result of the paper. \begin{shaded}
\begin{theorem}\label{one-point-WK}
Let $\tau$ be any  tau-function of the KdV hierarchy and let $\psi,\psi^*$ be defined by \eqref{Sato-tau}. The generating function of one-point correlation functions has the following expression
\be\label{one-point-generating}
F_1(z;\bdt)= \frac12  \, {\rm Tr} \le(\Psi^{-1}(z)\Psi_z(z)\sigma_3\ri)- \vartheta_z(z),
\ee
where
\be\label{psai}
\Psi(z;\bdt)=\left(
\begin{array}{cc}
\psi(z;\bdt) & \psi^*(z;\bdt) \\
-\psi_x(z;\bdt) & -\psi^*_x(z;\bdt)\\
\end{array}\right)
\ee
and $\sigma_3=\left(\begin{array}{cr} 1 & 0 \\ 0 & -1\end{array}\right)$ is the Pauli matrix.
\end{theorem}
\end{shaded}
\noindent The proof is in Sect.\,\ref{sect1point}. The identity \eqref{one-point-generating} gives a unique reconstruction, up to a multiplicative constant, of tau-function from the knowledge of the corresponding wave functions and thus provides a converse to Sato's result. 

\br
The reader conversant with the theory of isomonodromic deformation will observe similarity with the Jimbo--Miwa--Ueno definition of isomonodromic tau-function \cite{JMU,B}.
\er

\begin{remark}\label{uni-sato}
The theory developed by Sato et al. gives a  formal expansion of $\tau$ by using the infinite Grassmannian approach (see for example in \cite{IZ,BY,Zhou3} for details)
\be
\label{taugrass}
\tau=\sum_{\mu\in\mathbb{Y}} \pi_\mu \, \mathfrak{s}_\mu
\ee
which provides derivatives of $\tau$ at $\bdt=0$ in terms of Pl\"ucker coordinates. Here, $\mathbb{Y}$ denotes the set of Young diagrams, $\pi_\mu$ are Pl\"ucker coordinates of a point in the Sato Grassmannian,  $\mathfrak{s}_\mu=\mathfrak{s}_\mu(T)$ are Schur polynomials in $T=(T_1,T_2,\dots)$ with $t_{k}=-(2k+1)!! \, T_{2k+1}.$  However, the representation \eqref{taugrass} does not help much with computation of logarithmic derivatives of tau-function. \end{remark}
A well-known formula (see e.g. \cite{BBT}, pag. 389), expressing the generating function of  $\langle\langle\tau_0\tau_j\rangle\rangle,$ can be easily obtained as an immediate corollary of Theorem \ref{one-point-WK}:
\begin{corollary}\label{gen-ham} 
The following equality holds true:
\be\label{simple-two-point-generating}
1+\sum_{j=0}^\infty \frac{(2j+1)!!}{z^{2j+2}} \langle\langle\tau_0\tau_j\rangle\rangle= \psi(z)\, \psi^*(z).
\ee
\end{corollary}
\noindent It has been shown in \cite{Du0,BBT} that any pair of wave and dual wave functions $\psi,\psi^*$ satisfy
\be\label{bbt}
\psi(z)\,\psi^*(z)=1+\sum_{j=0}^\infty \frac{{\rm res}_{\p} L^{\frac{2j+1}{2}}}{z^{2j+2}}.
\ee
Hence from \eqref{bbt} and \eqref{simple-two-point-generating} we have 
$\langle\langle\tau_0\tau_j\rangle\rangle=\frac{1}{(2j+1)!!} {\rm res}_{\p} L^{\frac{2j+1}{2}},\, j\geq 0.$ 

Our next result expresses the generating function of two-point correlation functions in terms of wave and dual wave functions.
\begin{proposition}\label{two-point-WK}
Let $\tau$ be a tau-function of the KdV hierarchy and let $\psi,\psi^*$ defined by \eqref{Sato-tau}. The generating function of two-point correlation functions has the following expression
\bea\label{gen-two-point-functions}
F_2(z,w;\bdt)=\frac{\frac{1}{2} \Rsol_x(w)\Rsol_x(z)-\Rsol(w)\Rsol(z)\chi(z)\chi(-z)-\Rsol(z)\Rsol(w)\chi(w)\chi(-w)-(z^2+w^2)}{(z^2-w^2)^2},
\eea
where $\Rsol=\Rsol(z;\bdt)=\psi(z;\bdt)\,\psi^*(z;\bdt), \, \chi=\chi(z;\bdt)= \p_x \log \psi(z;\bdt)$; the explicit dependence on $\bdt$ of the functions $\Rsol(z;\bdt)$, $\Rsol(w;\bdt)$ etc. is omitted from \eqref{gen-two-point-functions}.
\end{proposition}
\noindent The proof is in Section \ref{two-point-section}. Note that both $\Rsol$ and $\chi$ are intended as   formal Laurent series in $z^{-1}$ whose coefficients are differential polynomials in $u$. 
The formula \eqref{gen-two-point-functions} was also derived in \cite{DZ-norm} through the Hadamard--Seeley expansion of $ {\rm e}^{\varepsilon\, L}$, $\varepsilon\to 0$.

Now we formulate the result for $n$-point correlation functions ($n\geq 2$); it includes the case $n=2$ but we have preferred to highlight it separately in the above Prop. \ref{two-point-WK} for clarity's sake.
\begin{shaded}
\begin{theorem}[Main Theorem] \label{multi-point}
Let $\tau$ be any tau-function of the KdV hierarchy and let $\psi = \psi(z;\bdt) ,\psi^*=\psi^*(z;\bdt)$ be defined by \eqref{Sato-tau}.  Let $\Rsol =\Rsol(z;\bdt)$ be as above. Denote $\Theta$ the matrix given by one of the following equivalent expressions
\bea
\Theta(z;\bdt)&=&\frac{1}{2} 
\left[
\begin{array}{cc}
-\Rsol_x & -2 \Rsol \\
\Rsol_{xx} -2(z^2-2\,u )\Rsol &  \Rsol_x \\
\end{array}
\right] \label{TH1}\\
&=&\frac{1}{2} 
\left[
\begin{array}{cc}
-(\psi\,\psi^*)_x & -2\, \psi\,\psi^*\\
2\,\psi_x\,\psi^*_x &  (\psi\,\psi^*)_x\\
\end{array}
\right] \label{TH11}\\
&=& z\,\Psi(z;\bdt)\, \sigma_3\, \Psi^{-1}(z;\bdt), \label{TH2}
\eea
where the matrix $\Psi(z; \bdt)$ is defined in \eqref{psai}.
Then the generating function of  $n$-point correlation functions with $n\geq 2$ has the following expression
\be
F_n(z_1,\dots,z_n;\bdt)=-\frac{1}{n}\sum_{r\in S_{n}}\frac{
 {\rm Tr} \le(\Theta(z_{r_1})\cdots \Theta(z_{r_{n}})\ri)
}{ \prod_{j=1}^n (z_{r_j}^2-z_{r_{j+1}}^2) }-\delta_{n,2}\frac{z_1^2+z_2^2}{(z_1^2-z_2^2)^2}. 
\label{n-point-gen}
\ee
Here $\delta_{n,2}$ is the Kronecker delta, $S_{n}$ denotes the group of permutations of $\{1,\dots,n\},$ and for a permutation $r=[r_1, \dots, r_n],$ $r_{n+1}$ is understood as $r_1$.
\end{theorem}
\end{shaded}
\noindent The proof is in Sect.\,\ref{sectnpoint}.
\br
The reader can verify  that the  right side of  \eqref{n-point-gen} is regular on the diagonals $z_i=z_j$.
\er
\begin{example}
The next example is $n=3$ for which the theorem above yields (on account of the cyclicity of the trace)
\bea
F_3(z_1,z_2,z_3;\bdt) = 
-\frac { {\rm Tr} \le(\Theta(z_1)\Theta(z_2)\Theta(z_3)\ri)-{\rm Tr} \le(\Theta(z_2)\Theta(z_1)\Theta(z_3)\ri)}{(z_1^2-z_2^2)(z_2^2-z_3^2)(z_3^2-z_1^2)}.
\eea
\end{example}
\br
We  have been recently made aware of an interesting preprint from Si-Qi Liu \cite{Liusq}, where the $n=2,3$ cases of Theorem \ref{multi-point} have also been obtained independently.  Jian Zhou also privately communicated a different nice proof of Theorem \ref{multi-point} \cite{Zhou}.
 \er

\paragraph{Application to intersection numbers of $\psi$-classes on $\overline{\mathcal M}_{g,n}$.} Recall that $\overline{\mathcal{M}}_{g,n}$ denotes the  Deligne--Mumford moduli space of stable curves of genus $g$ with $n$ marked points; let us also recall 

\noindent \textbf{Witten's conjecture (Kontsevich's theorem)}: {\it The partition function of $2D$ quantum gravity 
\be
\label{WittenKonts}
 Z(\bdt)=\exp\left(\sum_{n=0}^\infty \frac{1}{n!} \sum_{k_1,\dots,k_n\geq 0}  \langle \tau_{k_1}\cdots\tau_{k_n}\rangle\,  t_{k_1}\cdots t_{k_n}\right)
\ee 
is a tau-function of the KdV hierarchy \eqref{GD}. Here $\langle\tau_{k_1}\dots \tau_{k_n}\rangle$ are intersection numbers of $\psi$-classes over the Deligne--Mumford moduli spaces 
 \be\label{def-gwpt}
 \langle\tau_{k_1}\dots \tau_{k_n}\rangle:=\left\{
 \begin{array}{ll}
 \displaystyle 
 \int_{\overline{\mathcal{M}}_{g,n}} \psi_1^{k_1}\dots\psi_n^{k_n}, & \hbox{if }~g:=\frac{k_1+\dots+k_n-n+3}{3}\hbox{  is a nonnegative integer},\\
 0, & \hbox{otherwise}. \end{array}\right.
 \ee
In the above integrals, $\psi_i$ is the first Chern class of the $i$-{th} tautological line bundle over $\overline{\mathcal{M}}_{g,n}.$}

The partition function $Z(\bdt)$ is now generally  known as the {\it  Witten--Kontsevich tau-function}. Witten's conjecture implies that $u^{^{WK}}\!\!(\bdt):=\p_x^2 \log Z(\bdt)$ is a particular solution of \eqref{GD}, for which reason  we call it the {\it Witten--Kontsevich solution}, also known as {\it  the topological solution}. It can be specified by the initial data 
\be
u^{^{WK}}\!\!(\bdt)|_{t_0=x,~ t_{\geq 1}=0}=x.
\ee

Applying Theorem \ref{one-point-WK} and Theorem \ref{multi-point} to the Witten--Kontsevich tau-function, we obtain 
\begin{shaded}
\begin{theorem}\label{numbers} Let $M(z)$ denote the following matrix-valued formal series
\be\label{m-matrix}
M(z)=\frac{1}{2}\left(
\begin{array}{cc}
-\sum_{g=1}^\infty \frac{(6g-5)!!}{24^{g-1} \, (g-1)!} z^{-6g+4} & -2 \sum_{g=0}^\infty \frac{(6g-1)!!}{24^g \, g!} z^{-6g}\\
\\
2 \sum_{g=0}^\infty\frac{6g+1}{6g-1} \frac{(6g-1)!!}{24^g \, g!} z^{-6g+2} &  \sum_{g=1}^\infty \frac{(6g-5)!!}{24^{g-1} \, (g-1)!} 
z^{-6g+4}\\
\end{array}
\right).
\ee
The generating functions of $n$-point intersection numbers \eqref{def-gwpt}
 \be
 F_{n}^{^{WK}}\!\!(z_1,\dots,z_n):=\sum_{k_1,\dots,k_n=0}^\infty  \langle \tau_{k_1} \dots \tau_{k_n}\rangle\, \frac{(2k_1+1)!!}{z_1^{2k_1+2}}\cdots \frac{(2k_n+1)!!}{z_n^{2k_n+2}},\ \ n\geq 1
 \ee
are given by the following formul\ae:
\bea
&& F_{1}^{^{WK}}\!\!(z)=\sum_{g=1}^\infty  \frac{(6g-3)!!}{24^g \, g!} z^{-(6g-2)},\label{no-one}\\
&& F_n^{^{WK}}\!\!(z_1,\dots,z_n)=-\frac{1}{n}\sum_{r\in S_{n}}  \frac{{\rm Tr} \le( M(z_{r_1})\cdots M(z_{r_{n}})\ri)}{
\prod_{j=1}^n(z_{r_j}^2-z_{r_{j+1}}^2 )} -\delta_{n,2}\frac{z_1^2+z_2^2}{(z_1^2-z_2^2)^2}, \quad n\geq 2. \label{no-n}
\eea
\end{theorem}
\end{shaded}
\noindent The expression  \eqref{no-one} is equivalent to the well-known formula (for example see in \cite{Faber, Ok}):
\be \langle\tau_{3g-2}\rangle=\frac{1}{24^g \, g!},\quad g\geq 1.\ee
The formula \eqref{no-n} seems to be new.

\br
In \cite{Ok}, A.~Okounkov provided an expression for {\it  different} generating functions of the same intersection numbers. The relationship between the two approaches is the following
\bea
\begin{array}{c|c}
\hbox{Our approach} & \hbox{Okounkov's approach}\\
\hline
\ds  \sum_{k_1,\dots,k_n=0}^\infty  \frac{(2k_1+1)!!}{z_1^{2k_1+2}}\cdots \frac{(2k_n+1)!!}{z_n^{2k_n+2}}\,  \langle \tau_{k_1} \dots \tau_{k_n}\rangle
  &
\ds  \sum_{k_1,\dots,k_n=0}^\infty  x_1^{k_1}\cdots x_n^{k_n}\, \langle \tau_{k_1} \dots \tau_{k_n}\rangle
\end{array}
\eea
It should be apparent that the two generating functions are related by a (formal) multiple Laplace transform.
Okounkov's approach produced integral expressions (but reducing to explicit formul\ae\  for $n\leq 3$). On the other hand the formula \eqref{no-n} is absolutely explicit and allows an efficient computation of the n-point numbers also for very high genera.
See Section \ref{furtherrems} below for more details.
\er

\br
In \cite{BE} M.~Berg\`ere and B.~Eynard introduced correlators associated with the Christoffel--Darboux kernel obtained from solutions to a general $2\times 2$ system of linear ODEs with rational coefficients (see Definition 2.3 in \cite{BE}).  In this case one can use a suitably modified Jimbo--Miwa--Ueno formula \cite{JMU} as the definition of the tau-function. The construction of \cite{BE} was extended in \cite{BBE} to Christoffel--Darboux kernels associated with larger systems of ODEs.  
Applications of this approach to computation of intersection numbers on the Deligne--Mumford moduli spaces have not been discussed in refs. \cite{BE, BBE}.
We are grateful to B.~Eynard who,  after the first version of the present paper appeared on arXiv, kindly communicated us that, in the particular case of Airy kernel it can be shown that the Berg\`ere--Eynard generating functions of the correlators, after a suitable normalization, reproduce intersection numbers of $\psi$-classes. It is an interesting observation that certainly deserves a further study.
\er

\paragraph{Application to higher Weil--Petersson volumes.}
The {\it  Main Theorem} also allows us to compute higher Weil--Petersson volumes, 
by which we mean integrals of mixed $\psi$- and $\kappa$-classes of the form
\be\label{psi-kappa}
\langle \kappa_1^{d_1}\dots \kappa_\ell^{d_\ell} \tau_{k_1} \dots \tau_{k_n}\rangle_{g,n}:=\int_{\overline{\mathcal{M}}_{g,n}} \psi_1^{k_1}\dots\psi_n^{k_n}  \kappa_1^{d_1}\dots \kappa_\ell^{d_\ell}.
\ee
Recall \cite{AC, KK, Mumford} that $\kappa$-classes as elements of the Chow ring $A^*\left( \overline{\mathcal M}_{g,n}\right)$ are defined by
\be
\kappa_i =f_* \left(\psi_{n+1}^{i+1}\right)\in A^i\left( \overline{\mathcal M}_{g,n}\right), \quad i\geq 0,
\ee
where $f: \overline{\mathcal M}_{g,n+1}\to \overline{\mathcal M}_{g,n}$ is the universal curve (the forgetful map). The class $\kappa_0=2g-2+n$ is just a constant, so it will not appear below.
The integrals \eqref{psi-kappa} take zero values unless
\be
\sum_{j=1}^n k_j+\sum_{j=1}^l j\,d_j=3g-3+n.
\ee
We denote by $Z^{\kappa}(\bdt;\bds)$ the partition function of higher Weil--Petersson volumes
\be
Z^{\kappa}(\bdt;\bds)=\exp\left(\sum_{g,n,\ell\geq 0} \frac1{n!} \sum_{d_1,\dots,d_\ell,k_1,\dots,k_n\geq 0}   \langle \kappa_1^{d_1}\dots \kappa_\ell^{d_\ell} \tau_{k_1} \dots \tau_{k_n}\rangle_{g,n}\,  t_{k_1}\cdots t_{k_n} \frac{\fs_1^{d_1}\dots \fs_\ell^{d_\ell}}{d_1! \cdots d_\ell!}\right),
\ee
where $\bds$ denotes the infinite vector of independent variables $(\fs_1,\fs_2,\dots).$
It is a KdV tau-function of a family of solutions depending on the parameters $\bds$ \cite{KMZ}. The $\bds=\bdzero$ evaluation of this function gives the  Witten--Kontsevich tau-function:
$Z(\bdt)=Z^{\kappa}(\bdt;\bdzero).$ By $\psi^{^{WK}}\!\!(z;\bdt),\,\psi^{\kappa}(z;\bdt;\bds)$ denote the corresponding wave functions
\be
 \psi^{^{WK}}\!\!(z;\bdt):=\frac{Z(\bdt-[z^{-1}])}{Z(\bdt)} {\rm e}^{\vartheta(z;\bdt)},\quad \psi^\kappa(z;\bdt;\bds):=\frac{Z^\kappa(\bdt-[z^{-1}];\bds)}{Z^\kappa(\bdt;\bds)} {\rm e}^{\vartheta(z;\bdt)}.
\ee
\noindent{{\bf Notations.}} (i) $\mathbb{Y}$ will denote the set of all partitions. (ii) For a partition $\lambda=(\lambda_1\geq \lambda_2\geq \dots),$ we denote by $\ell(\lambda)=card\{i=1,2,\dots|\,\lambda_i\neq 0\}$ the {\it  length} of $\lambda$, by $|\lambda|=\sum_{i\geq 1}\lambda_i$ the {\it  weight} of $\lambda$,  by $m_i(\lambda)$ the {\it  multiplicity} of $i$ in $\lambda$, $i=1,\dots, \infty$. Let
\be
m(\lambda)!:=\prod_{i=1}^\infty m_i(\lambda)!.
\ee
The partition of $0$ is denoted by $(0); \, \ell((0))=|(0)|=0$.  For an arbitrary sequence of indeterminates $q_1,q_2,\dots$ denote $q_\lambda:=q_{\lambda_1}\cdots q_{\lambda_{\ell(\lambda)}};$ we use $\prec$ for the reverse lexicographic ordering on $\mathbb{Y},$ e.g. 
$ \mathbb{Y}_3=\{(3)\prec(2,1)\prec(1^4)\}, \, \mathbb{Y}_4=\{(4)\prec(3,1)\prec(2^2)\prec(2,1^2)\prec(1^4)\}.$

We give two ways of representing the generating function
of the intersection numbers $\langle \kappa_1^{d_1}\dots \kappa_\ell^{d_\ell} \tau_{k_1} \dots \tau_{k_n}\rangle$. The first one represents the generating function of intersections of $\psi$- and $\kappa$-classes in terms of the generating functions of intersections of $\psi$-classes obtained in Theorem \ref{numbers}. The second approach is based on an explicit representation of the wave function $\psi^\kappa(z;\bdt; \fs)$ in terms of the wave function of the Witten--Kontsevich solution. These two approaches are formulated as, respectively, the first and the second part of the following theorem.

\begin{shaded}

\begin{theorem} \label{WP-thm} 

\noindent {\it  Part I.}   For a given $\lambda\in\mathbb{Y}$ and for any $n\geq 0$
\bea
&& \quad \sum_{k_1,\dots,k_n\geq 0} 
\langle \kappa_{\lambda_1}\dots\kappa_{\lambda_{\ell(\lambda)}} \tau_{k_1}\dots \tau_{k_n}\rangle \frac{(2k_1+1)!!}{z_1^{2k_1+2}}\cdots \frac{(2k_n+1)!!}{z_n^{2k_n+2}}=(-1)^{\ell(\lambda)} \sum_{|\mu|=|\lambda|}  \frac{L_{\lambda\mu} }{m(\mu)!} \res{w_1=\infty} \dots\res{w_{\ell(\mu)}=\infty}\nn\\
&&  \frac{w_1^{2\mu_1+3}}{(2\mu_1+3)!!} \cdots \frac{w_{\ell(\mu)}^{2\mu_{\ell(\mu)}+3}}{(2\mu_{\ell(\mu)}+3)!!} \, F^{^{WK}}_{\ell(\mu)+n }  (w_1,\dots,w_{\ell(\mu)}, z_1,\dots,z_n)\,d w_1 \cdots d w_{\ell(\mu)}. \label{PII}
\eea
Here $L_{\lambda\mu}$ are transition matrices from the monomial basis of symmetric functions to the power sum basis \cite{Mac}, see Lemma \ref{Llm} below for an explicit expression for the entries of this matrix.

\noindent {\it  Part II.} The wave function corresponding to higher Weil--Petersson volume satisfies
\be\label{1.46}
\psi^\kappa(z;\bdt;\bds)=\exp\le(\sum_{k=1}^\infty \frac{h_{k}(-\bds)z^{2k+3}}{(2k+3)!!}\ri) \sum_{\lambda\in\mathbb{Y}} \, \frac{(-1)^{\ell(\lambda)} \fs_\lambda}{m(\lambda)!}  \sum_{|\mu|=|\lambda|} L_{\lambda\mu} \, \frac{(-1)^{\ell(\mu)}}{m(\mu)!}\p_{t_{\mu_1+1}}\dots\p_{t_{\mu_{\ell(\mu)}+1}}\psi^{^{WK}}\!\!(z;\bdt). \ee

Denote by $F^{\kappa}_{n} (z_1,\dots,z_n;\bds)$ the generating function of the higher Weil--Petersson volumes, i.e.
\be
F^{\kappa}_n (z_1,\dots,z_n;\bds):=\sum_{l\geq 0}\sum_{k_1,\dots,k_n\geq 0\atop d_1,\dots,d_l\geq 0}  \langle \kappa_1^{d_1}\dots \kappa_\ell^{d_\ell} \tau_{k_1} \dots \tau_{k_n}\rangle\, \frac{(2k_1+1)!!}{z_1^{2k_1+2}}\cdots \frac{(2k_n+1)!!}{z_n^{2k_n+2}}\frac{\fs_1^{d_1}\dots \fs_\ell^{d_\ell}}{d_1! \cdots d_\ell!}, \quad n\geq 1.
\ee
Then
\bea
&&F^{\kappa}_1(z;\bds)=\frac{1}{4\,z}\le(-A(z;\bds)\, B_z(-z;\bds)+B_z(z;\bds)\, A(-z;\bds)+B(z;\bds)\, A_z(-z;\bds)-A_z(z;\bds)\, B(-z;\bds)\ri),
\label{FSK1}
\\
&&
\label{FSK2} 
F^{\kappa}_n(z_1,\dots,z_n;\bds)=-\frac{1}{n}\sum_{r\in S_{n}}  \frac{{\rm Tr} \le( M^{\kappa}(z_{r_1};\bds)\cdots M^{\kappa}(z_{r_{n}};\bds)\ri)}{
\prod_{j=1}^n(z_{r_j}^2-z_{r_{j+1}}^2 )} -\delta_{n,2}\frac{z_1^2+z_2^2}{(z_1^2-z_2^2)^2}, \quad n\geq 2.
\ee
Here $A(z;\bds):=\psi^{\kappa}(z;\bdzero;\bds),\,B(z;\bds):=\psi_x^\kappa(z;\bdzero;\bds)$,
\be
M^{\kappa}(z;\bds):=\frac12\left(\begin{array}{lr} -\left[ A(z;\bds) B(-z;\bds) +A(-z;\bds)B(z;\bds)\right] & - 2A(z;\bds)A(-z;\bds)\\
\\ ~~~2 B(z;\bds) B(-z;\bds) & A(z;\bds) B(-z;\bds) +A(-z;\bds) B(z;\bds)\end{array}\right)
\ee
(cf. definition  \eqref{TH11} of the matrix $\Theta$).
\end{theorem}
\end{shaded}

The polynomials $h_k=h_k(x_1, \dots, x_k)$, $k\geq 1$ 
in the formula \eqref{1.46} are defined by the well known generating function
$$
1+\sum_{k\geq 1} h_k(x) z^k = {\rm e}^{\sum_{i\geq 1} x_i z^i}.
$$

\begin{example}
Using Thm.\,\ref{WP-thm} we obtain in particular
\bea
&&\langle\kappa_{3g-3}\rangle_{g,0}=\frac{1}{24^g\cdot g!},\qquad g\geq 1,\\
&&\langle\kappa_1 \tau_{3g-3}\rangle_{g,1}=3\frac{12g^2-12g+5}{5!! \cdot 24^g\cdot g!},\qquad g\geq 1,\\
&&\langle\kappa_2 \tau_{3g-4}\rangle_{g,1}=3\frac{72g^3-132g^2+95g-35}{7!!\cdot 24^g\cdot g!},\qquad g\geq 2,\\
&&\langle\kappa_3 \tau_{3g-5}\rangle_{g,1}=\frac{1296 g^4 - 3888 g^3+ 4482 g^2- 2835 g+945}{9!! \cdot 24^g\cdot g!},\qquad g\geq 2.
\eea
Moreover, we find 
\bea
&& A^{(1)}(z)=-\frac{z^{5}}{5!!}c+\frac{z^5}{5!!} q-\frac{z^2}{2\cdot 5!!}  c,\\
&& A^{(2)}(z)=-\frac{z^{7}}{7!!}c+\frac{8z^7+5z}{8\cdot 7!!} q-\frac{z^4}{2\cdot 7!!}  c,\\
&& A^{(1^2)}(z)=\left(\frac{z^{10}}{225}+\frac{11 z^7}{1575}-\frac{ z^4}{2520}\right)c+\le(-\frac{z^{10}}{225}-\frac{ z^7}{210}+\frac{3 z}{560}\ri)q,\\
&& B^{(1)}(z)=-\frac{z^{6}}{5!!}q+\frac{z^3}{2\cdot 5!!} q + \frac{4z^6-6}{4\cdot 5!!}  c,\\
&& B^{(2)}(z)=-\frac{z^{8}}{7!!}q+\frac{z^5}{2\cdot 7!!} q + \frac{8z^8-7z^2}{8\cdot 7!!}  c,\\
&& B^{(1^2)}(z)=\le(-\frac{z^{11}}{225}-\frac{z^8}{210}+\frac{z^5}{150}-\frac{z^2}{240}\ri) c+\left(\frac{ z^{11}}{225}+\frac{4  z^8}{1575}-\frac{13  z^5}{2520}\right)q.
\eea
Here $c=c(z), \,q=q(z)$ are Faber--Zagier series \eqref{def-c},\eqref{def-q}; for a partition $\lambda$, $A^{\lambda}(z)$ is the coefficient of $\fs_\lambda := \prod_{j= 1}^{\ell(\l)} \fs_{\l_j} $ in the $\bds$-expansion of $A(z;\bds),$ similarly for definition of $B^{\lambda}(z)$.
\end{example}

\paragraph{Organization of the paper}
In Sect.\,\ref{pre}, we briefly review the KdV hierarchy.
In Sect.\,\ref{sect1point}, we prove Theorem \ref{one-point-WK} and derive some useful lemmas. In Sect.\,\ref{sectnpoint}, we prove Prop.\,\ref{two-point-WK} and Theorem \ref{multi-point}. Application to the Witten--Kontsevich solution is presented in Sect.\,\ref{sectappl}. Application to higher Weil--Petersson volumes is presented in Sect.\,\ref{wp}. Further remarks are given in 
Sect.\,\ref{furtherrems}. We list useful formul\ae\ and tables of intersection numbers, 
and apply a generalization of the Kac--Schwarz operator  to higher Weil--Petersson volumes in appendices.

\section*{Acknowledgements}
We would like to thank Tamara Grava, Dmitry Korotkin, Youjin Zhang, Jian Zhou for helpful discussions. We are grateful to Erik Carlsson and Fernando Rodriguez Villegas for pointing out  that $L_{\lambda\mu}$ in Theorem \ref{WP-thm} are transition matrices from monomial basis to homogeneous basis of symmetric functions. 
We thank Bertrand Eynard for bringing our attention to the papers \cite{BE}, \cite{BBE}. 
D.Y. is grateful to Youjin Zhang for his advising. The intersection numbers that we computed in this paper have been verified for $3g-3 + n \leq 20$ with intersection numbers produced by Carel Faber's Maple program; we thank Faber for sharing it  publicly. The work is partially supported by PRIN 2010-11 Grant ``Geometric and analytic theory of Hamiltonian systems in finite and infinite dimensions" of Italian Ministry of Universities and Researches. 
M.B. is supported in part by the Natural Sciences and Engineering Research Council of Canada grant RGPIN/261229--2011 and by the FQRNT grant ‘‘Matrices Al\'eatoires, Processus Stochastiques et Syst\`emes Int\'egrables’’ (2013--PR--166790).

\section{A brief reminder of the KdV hierarchy}
\label{pre}
In this preliminary section we review some useful facts about the KdV hierarchy \eqref{GD}, and derive generalized Kac--Schwarz operators associated to string equations.

Let $\mathcal{A}_u$ be the ring of differential polynomials in $u$. Define a family of differential polynomials $\Omega_{p;0}(u;u_x,...,u_{p})$ with zero constant term for $p\geq 0$ through the Lenard--Magri recursion 
\bea
\label{Lenard} \pa_x \Omega_{p;0}&=& \frac{1}{2p+1}\le(2\, u\, \pa_x + u_x+\frac14 \pa_x^3\ri) \Omega_{p-1;0},\quad p\geq 1,\\
\Omega_{-1;0}&=&1.
\label{Lenard-1}
\eea
Here, $u_k:=\p_x^k u,\,k\geq 1$ are jet variables. Then the KdV hierarchy \eqref{GD} can be equivalently written as
\be
u_{t_{p}} = \pa_x \Omega_{p;0},\quad p\geq 0.
\label{KdV}
\ee
The first few $\Omega_{p;0}$ are given by
\bea
\Omega_{0;0}&=& u, \\
\Omega_{1;0}&=&\frac {u^2}2 + \frac {1} {12} u_{xx},\\
\Omega_{2;0}&=&\frac{u^3}{6}+\le(\frac{1}{12} u \,u_{xx}+\frac{1}{24} u_x^2\ri)+\frac{u_4}{240},\\
\Omega_{3;0}&=&\frac{u^4}{24}+ \le(\frac{1}{24} u^2 u_{xx}+\frac{1}{24} u\, u_x^2\ri)+\le(\frac{1}{240} u\,u_4+\frac{1}{120} u_xu_3+\frac{1}{160} u_{xx}^2\ri)
+\frac{u_6}{6720}.\eea
In general, we have \cite{DZ-norm,BBT,Du0}
\be \label{O-L}
\Omega_{p;0}=\frac{1}{(2p+1)!!} {\rm res}_{\p} L^{\frac{2p+1}{2}}, \quad \forall\,p\geq 0.
\ee

Let $\psi(z;\bdt)$ and $\psi^*(z;\bdt)$ be a pair of wave and dual wave functions of the KdV hierarchy.
As in the Introduction, we define the following two formal series in $z^{-1} $ by
\be\label{SR}
\Rsol(z;\bdt):=\psi(z;\bdt)\,\psi^*(z;\bdt),\quad \chi(z;\bdt):= \p_x \log \psi(z;\bdt).
\ee
We remind \cite{Dickii} that the function $\Rsol$ coincides with the diagonal value of the   {\it  resolvent} of the Lax operator $L$. One can easily see from \eqref{bbt} and \eqref{O-L} that it gives the generating series of $\Omega_{p;0}$, i.e.
\be\label{Rsolv-gen}
\Rsol(z;\bdt)=1+\sum_{k=0}^\infty \frac{(2k+1)!!}{z^{2k+2}}\Omega_{k;0}(u(\bdt);u_x(\bdt),\dots,u_k(\bdt)).
\ee
It is well known that the following two lemmas hold true for $\Rsol(z)$ and $\chi(z)$, for which we omit proofs.
\begin{lemma}\label{S-p}
The function $\Rsol=\Rsol(z;u;u_x,\dots)$ is the unique solution in $\mathcal{A}_u[[z^{-1}]]$ of the following ODE
\bea
&&\Rsol\,\Rsol_{xx}-\frac{1} 2 \, \Rsol_x^2+(4\,u-2\,z^2)\, \Rsol^2=-2\,z^2,\label{ODE-S}\\
&&\Rsol=1+\mathcal{O}(z^{-1}), \quad z\rightarrow \infty.\label{boundary-S}
\eea
Furthermore, the following formula holds true for $\Rsol(z)$
\be\label{disp-rsolv}
\Rsol(z;u,u_x,\dots)|_{u_j=0, \, j\geq 1}= \left(1-\frac{2\,u}{z^2}\ri)^{-\frac{1}{2}}.
\ee
In the above formul\ae\ it is understood that $\Rsol(z;\bdt)=\Rsol(z; u(\bdt);u_x(\bdt),\dots).$
\end{lemma}
\noindent One can recognize that  $\Rsol(z;u;u_x,\dots)$ is a full genus correction of gradient of the period of the one-dimensional Frobenius manifold \cite{Du1,DZ-norm}.

\begin{lemma} The function $\chi = \chi(z;\bdt)= \p_x \log \psi(z;\bdt)$ is the unique formal solution in $\mathcal{A}_u[[z^{-1}]]$ of the Riccati equation
\bea\label{Ric}
&&\chi_x+\chi^2+2\,u-z^2=0,\\
&&\chi(z;\bdt)=z+\sum_{k=1}^\infty \frac{\chi_k(\bdt)}{z^k},\quad z\rightarrow \infty.
\eea
\end{lemma}

\bl[\cite{BBT}] \label{wr}
Let $\psi = \psi(z;\bdt)$ and $\psi^* =\psi^*(z;\bdt)$ be a pair of wave and dual wave functions of the KdV hierarchy. The following formula holds true:
\bea
&&  \psi_x \, \psi^* - \psi\, \psi^*_x= 2\,z, \label{wronsk}\\
&& \chi(z;\bdt)=\frac{1} 2\, \p_x \log \Rsol(z;\bdt)+\frac{z}{\Rsol(z;\bdt)}. \label{wronsk-cor}
\eea
\el
\begin{proof}
The identity \eqref{wronsk} is well known. The proof of \eqref{wronsk-cor} is straightforward by using \eqref{wronsk} and \eqref{SR}.
\end{proof}

\section{One-point correlation functions of the KdV hierarchy}
\label{sect1point}
In this section we prove Theorem \ref{one-point-WK} by using the bilinear identities \eqref{bilinear-KdV}.
Following  \cite{BBT,DVV,DZ-norm}, we will frequently use  the following differential operator
\be
\label{defnabla}
\nabla (z):=\sum_{k=0}^\infty \frac{(2k+1)!!}{z^{2k+2}} \frac{\p}{\p t_k}.
\ee

\bl\label{cor-wr} For any tau-function $\tau(\bdt)$ of the KdV hierarchy and any $z\in\mathbb{C},\,z\neq 0$ we have
\bea
\frac{\tau_{x}(\bdt-[z^{-1}])\,\tau(\bdt+[z^{-1}])}{2\,z\,\tau(\bdt)^2}-\frac{\tau_{x}(\bdt+[z^{-1}])\,\tau(\bdt-[z^{-1}])}{2\,z\,\tau(\bdt)^2}+\frac{\tau(\bdt-[z^{-1}])\,\tau(\bdt+[z^{-1}])}{\tau(\bdt)^2}=1.
\eea
\el
\noindent The proof of this lemma is a straightforward application of Lemma \ref{wr} and hence omitted.

\noindent {\bf  Proof of Theorem \ref{one-point-WK}}. \quad 
Let  $\tau(\bdt)$ be  any tau-function of the KdV hierarchy. Let $\psi(z;\bdt)$ and $\psi^*(z;\bdt)$ be the corresponding wave and dual wave functions \eqref{Sato-tau}. The bilinear identities \eqref{bilinear-KdV},  \eqref{bilinear-psis} imply that $\forall \,\bdt,\wt\bdt, w$
\bea
&&
\oint_{R<|z|<|w|} \hspace{-25pt}\tau(\bdt-[w^{-1}]-[z^{-1}])\, \tau(\wt \bdt + [w^{-1}] + [z^{-1}])
  {\rm e}^{\sum_{j\geq 0} \le(t_{j}-\tilde t_{j}-2\,(2j-1)!!w^{-2j-1}\ri) \frac{z^{2j+1}}{ (2j+1)!!}}dz=0. \label{bi-ttw}
\eea
Here $R$ is a sufficiently large number. Note that
\be
 {\rm e}^{-2\sum_{j\geq 0} \frac{1}{2j+1}\le(\frac{z}{w}\ri)^{2j+1}}=\frac{w-z}{w+z}.
\ee
Applying the operator $\nabla(w)$ to  the above identity \eqref{bi-ttw} and
then setting  $\wt \bdt=\bdt$, we obtain that $\forall \,\bdt, w,$
\bea
&&
\sum_{k=0}^\infty \frac{(2k+1)!!}{w^{2k+2}}\le[\oint_{R<|z|<|w|} \le( \tau(\bdt-[w^{-1}]-[z^{-1}])\,\tau_{t_k} (\bdt + [w^{-1}] + [z^{-1}])\, \frac{w-z}{w+z}
\ri.\ri.\nn\\
&&\quad\le.\le.-\frac{z^{2k+1}}{(2k+1)!!}\, \tau(\bdt-[w^{-1}]-[z^{-1}])\, \tau( \bdt + [w^{-1}] + [z^{-1}])\, \frac{w-z}{w+z}
\ri)dz\ri] =0.
\label{middle}
\eea
Noticing now  that
\bea
\frac{w-z}{w+z}&=&\frac{2\,w}{z+w}-1,\quad z\rightarrow -w,\\
\frac{w-z}{w+z}&=&-1+2\,\frac{w}{z}+\mathcal{O}(z^{-2}),\quad z\rightarrow \infty
\eea
we have
\bea
&&\!\!\!\!\!\!~~\oint_{R<|z|<|w|} \tau(\bdt-[w^{-1}]-[z^{-1}])\, \tau_{t_k} (\bdt + [w^{-1}] + [z^{-1}])\, \frac{w-z}{w+z}\,\frac{dz}{2i\pi} \nn\\
&&\!\!\!\!\!\!=-2\,w\,\tau(\bdt)\,\tau_{t_k} (\bdt)+\oint_{R<|w|<|z|} \tau(\bdt-[w^{-1}]-[z^{-1}])\, \tau_{t_k} (\bdt + [w^{-1}] + [z^{-1}])\, \frac{w-z}{w+z}\, \frac{dz}{2i\pi}\nn\\
&&\!\!\!\!\!\!=-2\,w\,\tau(\bdt)\,\tau_{t_k} (\bdt)-\res{z=\infty} \le(\tau(\bdt-[w^{-1}])-\tau_{t_0}(\bdt-[w^{-1}])\,z^{-1} +\mathcal{O}(z^{-2}) \ri)\nn\\
&&\!\!\!\!\!\!\quad\quad \le(\tau_{t_k} (\bdt + [w^{-1}])+\tau_{t_0 t_k}(\bdt+[w^{-1}])\,z^{-1}+\mathcal{O}(z^{-2}) \ri)\, \le(-1+2\,w\,z^{-1}+\mathcal{O}(z^{-2}) \ri) \,dz\nn\\
&&\!\!\!\!\!\!=-2\,w\,\tau(\bdt)\,\tau_{t_k} (\bdt)+2\,w\,\tau(\bdt-[w^{-1}])\,\tau_{t_k}(\bdt+[w^{-1}]) \nn\\
&&\!\!\!\!\!\!\quad\quad\quad\quad+\tau_{t_0}(\bdt-[w^{-1}])\,\tau_{t_k}(\bdt+[w^{-1}])-\tau(\bdt-[w^{-1}])\,\tau_{t_0 t_k}(\bdt+[w^{-1}]). 
\label{A-res}
\eea
Similarly we have
\bea
&&\!\!\!\!\!\!\sum_{k=0}^\infty \frac{(2k+1)!!}{w^{2k+2}}\oint_{R<|z|<|w|}\frac{z^{2k+1}}{(2k+1)!!}\, \tau(\bdt-[w^{-1}]-[z^{-1}])\, \tau( \bdt + [w^{-1}] + [z^{-1}])\, \frac{w-z}{w+z}\,
\frac{dz}{2i\pi}\nn\\
&&\!\!\!\!\!\!\!\!\!\!=\oint_{R<|z|<|w|} \tau(\bdt-[w^{-1}]-[z^{-1}])\, \tau( \bdt + [w^{-1}] + [z^{-1}])\, \frac{z}{(w+z)^2}\,
\frac{dz}{2i\pi} \nn\\
&&\!\!\!\!\!\!\!\!\!\!=- \tau(\bdt)\, \tau(\bdt) 
+ \tau(\bdt - [w^{-1}]) \tau(\bdt+[w^{-1}]). 
\label{B-res}
\eea
Here we have used the following expansions
\bea
\frac{z}{(w+z)^2}&=&\frac{-w}{(z+w)^2}+\frac{1}{z+w},\quad z\rightarrow -w,\\
\frac{z}{(w+z)^2}&=&\frac{1}{z}+\mathcal{O}(z^{-2}),\quad z\rightarrow \infty.
\eea
Substituting \eqref{A-res} and \eqref{B-res} into \eqref{middle} we arrive at
\bea
&&\!\!\!\!\!\!\!\sum_{k=0}^\infty \frac{(2k+1)!!}{w^{2k+2}\,\tau(\bdt)^2}\left(\tau(\bdt-[w^{-1}])\,\tau_{t_k}(\bdt+[w^{-1}])+\frac{1}{2\,w}\,\tau_{t_0}(\bdt-[w^{-1}])\,\tau_{t_k}(\bdt+[w^{-1}]) \right.\nn\\
&&\left.-\frac{1}{2\,w}\,\tau(\bdt-[w^{-1}])\,\tau_{t_0 t_k}(\bdt+[w^{-1}])\right)=\sum_{k=0}^\infty \frac{(2k+1)!!}{w^{2k+2}} \frac{\tau_{t_k}(\bdt)}{\tau(\bdt)} - \frac1{2\,w}
+ \frac{\tau(\bdt - [w^{-1}]) \tau(\bdt+[w^{-1}])}{2\, w\, \tau(\bdt)^2}. \label{arrive}
\eea
Replacing $w\rightarrow -w$ in \eqref{arrive} we obtain
\bea
&&\!\!\!\!\!\!\!\sum_{k=0}^\infty \frac{(2k+1)!!}{w^{2k+2}\,\tau(\bdt)^2}\left(\tau(\bdt+[w^{-1}])\,\tau_{t_k}(\bdt-[w^{-1}])-\frac{1}{2\,w}\tau_{t_0}(\bdt+[w^{-1}])\,\tau_{t_k}(\bdt-[w^{-1}]) \right.\nn\\
&&\left.+\frac{1}{2\,w}\tau(\bdt+[w^{-1}])\,\tau_{t_0 t_k}(\bdt-[w^{-1}])\right)=\sum_{k=0}^\infty \frac{(2k+1)!!}{w^{2k+2}} \frac{\tau_{t_k}(\bdt)}{\tau(\bdt)} + \frac1{2\,w}
- \frac{\tau(\bdt - [w^{-1}]) \tau(\bdt+[w^{-1}])}{2\, w\, \tau(\bdt)^2} . \label{arrive-mirror}
\eea

Now let us look at the r.h.s. of \eqref{one-point-generating}. It is straightforward to compute the following products
\bea
 \psi_x\psi^*_z&=&-\frac{\tau_{t_0}(\bdt-[z^{-1}])}{\tau(\bdt)^2}\sum_{k=0}^\infty \tau_{t_k}(\bdt+[z^{-1}])\, \frac{(2k+1)!!}{z^{2k+2}}-\frac{\tau_{t_0}(\bdt-[z^{-1}])\tau(\bdt+[z^{-1}])}{\tau(\bdt)^2}\vartheta_z\nn\\
&&-z\,\frac{\psi(z;\bdt)}{\tau(\bdt)}\sum_{k=0}^\infty \tau_{t_k}(\bdt+[z^{-1}]) \frac{(2k+1)!!}{z^{2k+2}}  {\rm e}^{-\vartheta}-z\,\vartheta_z\,\psi(z;\bdt)\psi^*(z;\bdt)\nn\\
&&+\frac{\psi(z;\bdt) \tau_{t_0}(\bdt)}{\tau(\bdt)^2}\sum_{k=0}^\infty \tau_{t_k}(\bdt+[z^{-1}]) \frac{(2k+1)!!}{z^{2k+2}}  {\rm e}^{-\vartheta} + \frac{\tau_{t_0}(\bdt)}{\tau(\bdt)}\,\vartheta_z\,\psi(z;\bdt)\psi^*(z;\bdt),\\
 \psi_z\psi^*_x &=&-\psi_x\psi^*_z|_{z\rightarrow -z},\\
\psi\psi^*_{xz}&=& - \sum_{k=0}^\infty \frac{(2k+1)!!}{z^{2k+2}\,\tau(\bdt)^2}\,\tau(\bdt-[z^{-1}]) \tau_{t_0 t_k}(\bdt+[z^{-1}])-\frac{\tau(\bdt-[z^{-1}])\,\tau_{t_0}(\bdt+[z^{-1}])}{\tau(\bdt)^2}\vartheta_z\nn\\
&&-\psi(z;\bdt)\,\psi^*(z;\bdt)-\le(z\,\psi(z;\bdt)+\frac{\tau_{t_0}(\bdt)}{\tau(\bdt)}\, \psi(z;\bdt)\ri)\psi^*_z(z;\bdt),\\
\psi_{xz}\psi^*&=&-\psi\psi^*_{xz}|_{z\rightarrow-z}.
\eea
As a result we have
\bea
&\&\hspace{-15pt}\frac{\psi_x \psi^*_w+\psi_w \psi^*_x-\psi \psi^*_{xw}-\psi_{xw} \psi^*}{4\,w}=\nn\\
&&\!\!\!\!\!\!\!=\frac{1}{4\,w} 
\sum_{k=0}^\infty \frac{(2k+1)!!}{w^{2k+2}\tau(\bdt)^2}\le[\tau(\bdt-[w^{-1}])\,\tau_{t_0 t_k}(\bdt+[w^{-1}])-\tau_{t_0}(\bdt-[w^{-1}])\,\tau_{t_k}(\bdt+[w^{-1}])\ri.\nn\\
&&\!\!\!\!\!\!\!\quad\quad\quad\le.-\tau(\bdt+[w^{-1}])\,\tau_{t_0 t_k}(\bdt-[w^{-1}])+\tau_{t_0}(\bdt+[w^{-1}])\,\tau_{t_k}(\bdt-[w^{-1}])\ri]\nn\\
&&\!\!\!\!\!\!\!\quad-\frac{1} 2 \sum_{k=0}^\infty \frac{(2k+1)!!}{w^{2k+2}\tau(\bdt)^2} \le[\tau(\bdt-[w^{-1}])\,\tau_{t_k}(\bdt+[w^{-1}])+\tau(\bdt+[w^{-1}])\,\tau_{t_k}(\bdt-[w^{-1}])\ri]\nn\\
&&\!\!\!\!\!\!\!\quad+\,\frac{\vartheta_w}{2\,\tau(\bdt)^2}\, \left[\frac{1} w\,\tau(\bdt-[w^{-1}])\, \tau_{t_0} (\bdt+[w^{-1}])-\frac{1} w\,\tau(\bdt+[w^{-1}])\, \tau_{t_0} (\bdt-[w^{-1}])
-2 \,\tau(\bdt-[w^{-1}])\,\tau(\bdt+[w^{-1}])\right]\nn\\
&&\!\!\!\!\!\!\!=-\sum_{k=0}^\infty \frac{(2k+1)!!}{w^{2k+2}} \frac{\tau_{t_k}(\bdt)}{\tau(\bdt)}-\,\vartheta_w.
\eea
The last equality uses \eqref{arrive},\,\eqref{arrive-mirror} and Lemma \ref{cor-wr}. The theorem is proved.
\hfill 
\QED

\section{Multi-point correlation functions of the KdV hierarchy}
\label{sectnpoint}
In this section we apply Theorem \ref{one-point-WK} to derive generating functions of multi-point correlation functions of the KdV hierarchy. We will need the following formulae for the time derivatives of wave functions.
\bl\label{fundamental} 
Let $\psi(z;\bdt)$ and $\psi^*(z;\bdt)$ be any pair of wave and dual wave functions and let $\Rsol(z;\bdt)=\psi(z;\bdt)\,\psi^*(z;\bdt).$ It follows that
\bea
\nabla (z) \, \psi(w)&=&\frac{2\,\Rsol(z)\psi_x(w)-\Rsol_x(z) \psi(w)}{2\,(z^2-w^2)},\label{1-bbt}\\
\nabla (z) \, \psi_x(w)&=&\frac{ \Rsol_x(z)\psi_x(w)-\le[\Rsol_{xx}(z)-2\,\Rsol(z)(w^2-2\,u)\ri] \psi(w)}{2\,(z^2-w^2)},\\
\nabla (z) \,  \psi_w(w)&=&\frac{2\,\Rsol(z)\psi_{xw}(w)-\Rsol_x(z) \psi_w(w)}{2\,(z^2-w^2)}- w\frac{2\,\Rsol(z)\psi_x(w)-\Rsol_x(z) \psi(w)}{(z^2-w^2)^2},\\
\nabla (z) \,  \Psi(w)&=&\frac1{2(z^2-w^2)}
\left(
\begin{array}{cc}
- \Rsol_x(z) & -2\,\Rsol(z)\\
 \Rsol_{xx}(z)-2\,(w^2-2\,u)\Rsol(z) &  \Rsol_x(z)\\
\end{array}
\right)\Psi(w),\\
\nabla (z) \, \Psi(w)&=&\frac{1}{z^2-w^2}\Theta(z) \Psi(w) +
\left(
\begin{array}{cc}
0 & 0\\
\Rsol(z) & 0\\
\end{array}
\right)\Psi(w).\label{fundamental-imp}
\eea
\el
\begin{proof}
The identity \eqref{1-bbt} is a standard result; e.g. see the formula (3.12) in \cite{Du0} or (11.24) in \cite{BBT}. The others follow from \eqref{1-bbt}.
\end{proof}
\br\label{ana-rmk}
The Definition of the matrices $\Psi$ and $\Theta$ is valid not just in a formal sense. In many cases the matrix $\Psi$ solves a {\it  Riemann--Hilbert problem}; for example if the potential $u(x)$ is in an appropriate Schwartz class or it belongs to the class of quasi-periodic functions (as in the {\it  finite gap integration} problems). In these cases the identities for generating functions in the next section have also an {\it  analytic} meaning and allow  to tackle the problem of studying large time behaviours. 
This is interesting e.g. in the case of the particular solution of Witten--Kontsevich (see below), because the matrix $\Psi$ could be written not just as a formal series but analytically  in terms of Airy functions. 
We postpone this type of analysis to a forthcoming publication \cite{BBBY}.
\er
\br
The matrix  $\Theta$ is a resolvent of the matrix-valued Lax operator; see e.g. pages 158--159 in \cite{Dickii}. 
\er
\subsection{Generating function of two-point correlation functions} \label{two-point-section}
In this subsection we prove Prop. \ref{two-point-WK}.

\noindent {\it  Proof} of Prop. \ref{two-point-WK}. \qquad
Applying the operator $\nabla (w)$ on both sides of \eqref{one-point-generating} we obtain
\bea
&&~~\sum_{k=0}^\infty \frac{(2k+1)!!}{w^{2k+2}}\sum_{j=0}^\infty \frac{(2j+1)!!}{z^{2j+2}} \langle\langle\tau_k\tau_j\rangle\rangle\
=\frac12\, \nabla (w) \, {\rm Tr} \le(\Psi^{-1}(z)\Psi_z(z)\sigma_3\ri)- \nabla (w)\vartheta_z(z).
\eea
Clearly
\be
\nabla (w) \vartheta_z(z)=\frac{z^2+w^2}{(z^2-w^2)^2}.
\ee
It is easy to see that
$
\Psi^{-1}(z)=\frac{1}{2\,z}\sigma_2 \Psi(z)^T\sigma_2.
$
Here and below we will use Pauli matrices\be
\sigma_1=\le(\begin{array}{cc}
0 & 1\\
1 & 0\\
\end{array}\ri); \quad \sigma_2=\le(\begin{array}{cc}
0 & -i\\
i & 0\\
\end{array}\ri); \quad \sigma_3=\le(\begin{array}{cc}
1 & 0\\
0 & -1\\
\end{array}\ri).
\ee
\label{defPsi}
We have
\bea
&&\!\!\!\!\!\!\nabla (w) \, {\rm Tr} \le(\Psi^{-1}(z)\Psi_z(z)\sigma_3\ri)= \nn\\
&&\!\!\!\!\!\!=\frac{1}{2\,z}{\rm Tr} \left[\sigma_2 \nabla (w)\le(\Psi(z)^T\ri) \sigma_2 \Psi_z(z)\sigma_3+\sigma_2 \Psi(z)^T\sigma_2 \nabla (w) \le(\Psi_z(z)\ri)\sigma_3\right]\nn\\
&&\!\!\!\!\!\!=\frac{2}{(w^2-z^2)^2}{\rm Tr}\le(\Theta(w)\Theta(z)\ri)-\frac{i\,{\rm Tr}\le[\Psi(z)^T(\Theta(w)^T\,\sigma_2+\sigma_2\,\Theta(w))\Psi_z(z)\sigma_1\ri] 
}{2\,z\,(w^2-z^2)}\nn\\
&&\qquad\qquad\qquad-\frac{i}{2\,z} \, {\rm Tr}\le[\Psi(z)^T (Q(w)^T\,\sigma_2+\sigma_2 \,Q(w))\Psi_z(z)\sigma_1\ri].
\eea
Here 
\be
\label{defQ}
Q(w):=\left(
\begin{array}{cc}
0 & 0\\
\Rsol(w) & 0\\
\end{array}
\right).
\ee
 Noticing that
$
Q(w)^T\,\sigma_2+\sigma_2 \,Q(w)=0,\ 
 \Theta(w)^T\,\sigma_2+\sigma_2\,\Theta(w)=0
$, 
we have
\be
\nabla (w) \le[ {\rm Tr} \le(\Psi^{-1}(z)\Psi_z(z)\sigma_3\ri)\ri]=\frac{2}{(w^2-z^2)^2}{\rm Tr}\le(\Theta(w)\Theta(z)\ri).
\ee
Hence
\bea
&&~~\!\!\!\!\!\!\sum_{k=0}^\infty \frac{(2k+1)!!}{w^{2k+2}}\sum_{j=0}^\infty \frac{(2j+1)!!}{z^{2j+2}} \langle\langle\tau_k\tau_j\rangle\rangle
=\frac{{\rm Tr}\le(\Theta(w)\Theta(z)\ri)-z^2-w^2}{(w^2-z^2)^2}\\
&&\!\!\!\!\!\!=\frac{\frac{1}{2}\Rsol_x(w)\Rsol_x(z)-\frac{1}{2}(\Rsol_{xx}(w)\Rsol(z)+\Rsol_{xx}(z)\Rsol(w))+(z^2+w^2-4\,u)\Rsol(w)\Rsol(z)-z^2-w^2}{(w^2-z^2)^2}.\nn\\
\eea
Substituting \eqref{ODE-S} into the formula we obtain \eqref{gen-two-point-functions}.
The theorem is proved.
\hfill
\QED

\begin{example}
The first few two-point correlation functions are given by
\bea
&&\lla\tau_0^2\rra=u,\quad \lla\tau_0\tau_1\rra=\frac{u^2}{2}+\frac{1}{12}u_{xx},\\
&&\lla\tau_1^2\rra=\frac{u^3}{3}+ \left(\frac{u_x^2}{24} +\frac{u\, u_{xx}}{6} \right)+\frac{1}{144}u_{xxxx},\\
&&\lla\tau_1\tau_2\rra=\frac{u^4}{8}+ \le(\frac{u\, u_x^2}{12}+\frac{1}{8} u^2 u_{xx}\ri)+\le(\frac{1}{90}u\, u_4+\frac{23\, u_{xx}^2}{1440}+\frac{1}{60} \, u_x u_3\ri)+\frac{1}{2880}u_6.
\eea
\end{example}
To end this subsection we point out that, for an arbitrary solution $u(\bdt)$ of the KdV hierarchy, the Sato tau-function 
$\tau$ and the axiomatic tau-function $\tau^{KdV}$ defined in \cite{DZ-norm} coincide up to a gauge factor of the form $ {\rm e}^{\alpha_{-1}+\sum_{j\geq 0} t_{j} \alpha_{j}}$. Indeed, by definition we know that $\forall\,p,q,r,$
\be
\lla\tau_p\tau_q\rra=\lla\tau_q\tau_p\rra, \quad \lla\tau_p\tau_q\tau_r\rra=\lla\tau_q\tau_r\tau_p\rra=\lla\tau_r\tau_p\tau_q\rra.
\ee
From Proposition~\ref{two-point-WK} and Lemma \ref{S-p} we know that $\lla\tau_p\tau_q\rra \in \mathcal{A}_u$ and that 
\be
\lla\tau_p\tau_q\rra|_{u_j=0,\,j\geq 1}=\frac{u^{p+q+1}}{p!q!(p+q+1)}.
\ee
So $\tau$ satisfies the defining properties of $\tau^{KdV}$.

\subsection{Generating functions of multi-point correlation functions}

\noindent {\bf Notation.} For a permutation $r\in S_n,\,n\geq 2$ denote 
\be P(r)= -\prod_{j=1}^n \frac{1}{z_{r_j}^2-z_{r_{j+1}}^2},\ee
where $r_{n+1}$ is understood as $r_1.$

\bl The following formula holds true:
\bea\label{gen-Theta}
\nabla (z) \Theta(w)=\frac{1}{z^2-w^2}[\Theta(z) ,\Theta(w)] + [Q(z),\Theta(w)].
\eea
with $Q$ defined  in \eqref{defQ}.
\el
\begin{proof} We have, using the  Leibniz rule,
\bea
\nabla (z) \, \Theta(w)&=&\nabla (z) \left[\Psi(w)\sigma_3\Psi(w)^{-1}\right]=\nabla (z)\le(\Psi(w)\ri)\sigma_3 \Psi(w)^{-1}+\Psi(w)\sigma_3\nabla (z)\left(\Psi(w)^{-1}\right)\nn\\
&=&\left[\frac{1}{z^2-w^2}\Theta(z) \Psi(w) +
Q(z)\Psi(w)\right]\sigma_3 \Psi(w)^{-1}\nn\\
&&\qquad+\Psi(w)\sigma_3\left[\frac{1}{z^2-w^2}\Psi(w)^{-1}\sigma_2\Theta(z)^T\sigma_2+\Psi(w)^{-1}\sigma_2Q(z)^T\sigma_2\right]\nn\\
&=&\frac{1}{z^2-w^2}[\Theta(z) ,\Theta(w)] + [Q(z),\Theta(w)].
\eea
\end{proof}

\noindent {\it  Proof} of Theorem \ref{multi-point}. We use mathematical induction on $n$.
For $n=2,$ the theorem has been already verified in  Prop. \ref{two-point-WK}. Suppose it is true for $n=p\,(p\geq 2).$ For $n=p+1,$ 
\bea
&&~ \sum_{k_1,\dots,k_{p+1}=0}^\infty  \frac{(2k_1+1)!!}{z_1^{2k_1+2}}\cdots \frac{(2k_{p+1}+1)!!}{z_{p+1}^{2k_{p+1}+2}}\langle \langle \tau_{k_1} \dots \tau_{k_{p+1}}\rangle\rangle\nn\\
&&=\nabla (z_{p+1}) \sum_{k_1,\dots,k_{p}=0}^\infty  \frac{(2k_1+1)!!}{z_1^{2k_1+2}}\cdots \frac{(2k_{p}+1)!!}{z_{p+1}^{2k_{p}+2}}\langle \langle \tau_{k_1} \dots \tau_{k_{p}}\rangle\rangle\nn\\
&&=\nabla (z_{p+1}) \left(\frac{1}{p}\sum_{r\in S_{p}} P\le(r\ri) {\rm Tr} \le(\Theta(z_{r_1})\cdots \Theta(z_{r_{p}})\ri)-\delta_{p,2}\frac{z_1^2+z_2^2}{(z_1^2-z_2^2)^2}\right)\nn
\\
&&=\frac{1}{p} \sum_{r\in S_p} P(r) \sum_{j=1}^p  {\rm Tr} \le[\Theta(z_{r_1})\cdots \nabla (z_{p+1})\!\le(\Theta(z_{r_j})\ri) \cdots\Theta(z_{r_{p}})\ri]\nn\\
&&=\frac1p \sum_{r\in S_p} P(r) \sum_{j=1}^p  {\rm Tr} \le[\Theta(z_{r_1})\cdots \le(\frac{1}{z_{p+1}^2-z_{r_j}^2}[\Theta(z_{p+1}) ,\Theta(z_{r_j})] + [Q(z_{p+1}),\Theta(z_{r_j})]\ri) \cdots\Theta(z_{r_{p}})\ri]\nn\\
&&{\mathop{=}^{\star}}\frac1p\sum_{r\in S_p} P(r) \sum_{j=1}^p \left( \frac{1} {z_{p+1}^2-z_{r_j}^2} - \frac{1} {z_{p+1}^2-z_{r_{j-1}}^2} \right) {\rm Tr} \le(\Theta(z_{r_1})\cdots \Theta(z_{r_{j-1}}) \Theta(z_{p+1}) \Theta(z_{r_{j}}) \cdots\Theta(z_{r_{p}})\ri)\nn
\eea
\bea
&& = \frac1p  \sum_{j=1}^p \sum_{r\in S_p} P(r) \frac{z_{r_j}^2-z_{r_{j-1}}^2} {(z_{p+1}^2-z_{r_j}^2)(z_{p+1}^2-z_{r_{j-1}}^2)} {\rm Tr} \le(\Theta(z_{p+1}) \Theta(z_{r_{j}}) \cdots\Theta(z_{r_{p}})\Theta(z_{r_1})\cdots \Theta(z_{r_{j-1}})\ri)\nn\\
&&=\frac1p  \sum_{j=1}^p \sum_{r\in S_p} P\le([p+1,r_j,\dots,r_p,r_1,\dots,r_{j-1}]\ri) {\rm Tr} \le(\Theta(z_{p+1}) \Theta(z_{r_{j}}) \cdots\Theta(z_{r_{p}})\Theta(z_{r_1})\cdots \Theta(z_{r_{j-1}})\ri)\nn\\
&&\mathop{=}^{\dagger}\sum_{r\in S_{p}} P\le([p+1,r]\ri) \,  {\rm Tr}  \le(\Theta(z_{p+1})\Theta(z_{r_1})\cdots \Theta(z_{r_{p}})\ri)\nn\\
&&=\frac{1}{p+1}\sum_{r\in S_{p+1}} P(r) \, {\rm Tr}  \le(\Theta(z_{r_1})\cdots \Theta(z_{r_{p+1}})\ri).
\eea
Here, for a permutation $r=[r_1,\dots,r_\ell]\in S_\ell,$ $r_0:=r_\ell,$ and we have used the facts that both $P(r)$ and matrix trace of products are invariant under the cyclic permutation. 
In the step marked with an asterisk, we have used that  the terms of the form 
\be
 \sum_{j=1}^p  {\rm Tr} \bigg[\Theta(z_{r_1})\cdots 
 \overbrace{[Q(z_{p+1}),\Theta(z_{r_j})]}^{\hbox{$j$--th place}} \cdots\Theta(z_{r_{p}})\bigg]
\ee
are zero; to see this the reader can notice that the above is the infinitesimal action of the conjugation action by ${
\rm e}^{s Q(z_{p+1})}$ and the trace ${\rm Tr} (\Theta(z_{r_1})\cdots \Theta(z_{r_p}))$ is invariant under simultaneous conjugation of the matrices. 
In the step marked with $\dagger$ we have relabeled the summation over permutations. In the last step we have used that the summand is invariant under cyclic permutations of the indices.
The theorem is proved.
\hfill \QED
%

\section{On string equation and Kac--Schwarz operator}\label{KaSch}
In the remaining part of this paper we will consider a subset of solutions of the KdV hierarchy \cite{DZ-norm} whose tau-functions are specified by string equations:
\bea\label{string-general}
&&L_{-1} \tau=0,\qquad L_{-1}:=\sum_{k\geq 0} \tilde{t}_{k+1} \frac{\p}{\p t_{k}}+\frac1{2} \tilde t_0^{\,2},
\ee
where $\tilde t_{k}=t_{k}-c_{k},\,\, k\geq 0,$ with arbitrary constants $c_k$. Note that the string equation \eqref{string-general} uniquely determines a tau-function 
up to a multiplicative constant independent of $\bdt$. The Witten--Kontsevich tau-function and the generating function of higher Weil--Petersson volumes are particular examples in this class of solutions; e.g., for the Witten--Kontsevich solution $c_1=1$, $c_k=0$ for $k\neq 1$. 

It was observed by V.~Kac and A.S.~Schwarz \cite{KS} that the string equation for the Witten--Kontsevich tau-function implies a differential equation in the spectral parameter $z$ for the associated wave function. The equation can be written in terms of the linear differential operator $S_z=\frac{1}{z}\p_z-\frac1{2z^2}-z$ that will be called {\it Kac--Schwarz operator}. In this section,
following A.~Buryak \cite{Buryak}, we give a generalization of the Kac--Schwarz operator for the general class of solutions satisfying \eqref{string-general}.

\begin{definition}\label{def-KS}
The generalized Kac--Schwarz operator $S_z$ associated to the string equation \eqref{string-general} is defined as the following differential operator of $z:$
\be
S_z=\frac{1}{z}\p_z-\frac{1}{2z^2} -\sum_{k=0}^\infty \frac{c_{k}}{(2k-1)!!} z^{2k-1},
\ee
where the constants $c_k$ are the shifts in the times $t_k$ in \eqref{string-general}.
\end{definition}
 The above definition also can been generalized to  the Gelfand--Dickey hierarchy as in \cite{KS,BeY}.
\begin{proposition}  \label{KS-p} 
Let $K(z;\bdt):=\psi(z;\bdt)\cdot \tau(\bdt) = \tau(\bdt - [z^{-1}]) \,  {\rm e}^{\vartheta(z)}$. Then the following formula holds true:
\be \label{KS-op} L_{-1} K=S_z K. \ee
\end{proposition}
\begin{proof} On one hand we have
\bea
L_{-1} K&=&L_{-1} \le(\tau(\bdt-[z^{-1}])\cdot  {\rm e}^\vartheta\ri)\nn\\
&=&  {\rm e}^\vartheta \cdot \sum_{k=0}^\infty \tilde{t}_{k+1} \frac{\p  \tau(\bdt-[z^{-1}])}{\p t_k}+  \sum_{k\geq 0} \frac{\tilde{t}_{k+1} \, z^{2k+1}}{(2k+1)!!}  K +\frac{(\tilde t_0)^2}{2} K\nn\\
&=&  {\rm e}^\vartheta \cdot \sum_{k=0}^\infty \left(\tilde{t}_{k+1}- \frac{(2k+1)!!}{z^{2k+3}}\right) \frac{\p  \tau(\bdt-[z^{-1}])}{\p t_k}+ \sum_{k\geq 0} \frac{\tilde{t}_{k+1} \,  z^{2k+1}}{(2k+1)!!} K +\frac{(\tilde t_0-z^{-1})^2}{2} K\nn\\
&&+\, {\rm e}^\vartheta \cdot \sum_{k=0}^\infty \frac{(2k+1)!!}{z^{2k+3}}\frac{\p  \tau(\bdt-[z^{-1}])}{\p t_k}-\frac{1}{2z^{2}}K+\frac{1}{z} \tilde{t}_0 K\nn\\
&=& \, {\rm e}^\vartheta \cdot \sum_{k=0}^\infty \frac{(2k+1)!!}{z^{2k+3}}\frac{\p  \tau(\bdt-[z^{-1}])}{\p t_k}-\frac{1}{2z^{2}}K+ \sum_{k\geq -1} \frac{\tilde{t}_{k+1} \,  z^{2k+1}}{(2k+1)!!}K.
\eea
On another hand we have
\bea
\frac{1}{z} \p_z K =  {\rm e}^{\vartheta} \sum_{k=0}^\infty  \frac{(2k+1)!!}{z^{2k+3}}\frac{\p  \tau(\bdt-[z^{-1}]) }{\p t_k}+\sum_{k\geq 0} \frac{ t_k z^{2k-1} }{ (2k-1)!!}K
\eea
Since $\wt t_k = t_k - c_k$, we obtain
$
L_{-1} K=\frac{1}{z}\p_z K -\frac{1}{2z^2} K-\sum_{k=0}^\infty \frac{c_{k} z^{2k-1} }{ (2k-1)!!}K
$, which was our statement.
\end{proof}

\begin{proposition} \label{cor-string}
Let $u_0(x):=\p_x^2 \log \tau(x,\bdzero)$ and let $f(z;x):=\psi(z;x,\bdzero)$; then
\bea
\label{S1} && S_z f(z;x)=-c_1 f_x(z;x) - \sum_{k\geq 1} c_{k+1} \psi_{t_k}(z; x,\bdzero),\\
\label{S1-x}
&& S_z f_x(z;x)=-c_1 \le(z^2-2u_0(x)\ri) f(z;x) - \sum_{k\geq 1} c_{k+1} \psi_{ t_0 t_k}(z; x,\bdzero),\\
\label{S3} && \psi_{t_k}(z; x,\bdzero)=\frac{1}{(2k+1)!!} \sum_{i=0}^k (2i-1)!!\, z^{2k-2i}\le(\Omega_{i-1;0}|_{u\rightarrow u_0} f_x(z;x) - \frac12 f(z;x)\p_x\Omega_{i-1;0}|_{u\rightarrow u_0} \ri),\\
\label{S4} && \psi_{t_0 t_k}(z; x,\bdzero)=\frac{1}{(2k+1)!!} \sum_{i=0}^k (2i-1)!! \, z^{2k-2i}\le(\p_x\Omega_{i-1;0}|_{u\rightarrow u_0} f_x(z;x)+\Omega_{i-1;0}|_{u\rightarrow u_0} \le(z^2-2u_0(x)\ri) f(z;x)\ri.\nn\\
&&\le.\qquad\qquad - \frac12 f_x(z;x)\p_x\Omega_{i-1;0}|_{u\rightarrow u_0}- \frac12 f(z;x)\p_x^2\Omega_{i-1;0}|_{u\rightarrow u_0} \ri).
\eea
\end{proposition}
\begin{proof}
Taking $t_k=0,\,k\geq 1$ in \eqref{KS-op} gives
\be S_z f(z;x)=\frac{c_0^2}{2} f(z;x)-\sum_{k\geq 0} c_{k+1} \frac{\p \psi}{\p t_k}(z; x,\bdzero)-\sum_{k\geq 0} c_{k+1}\frac{\p \log \tau}{\p t_k}(x,\bdzero) f(z;x). \ee
Taking $t_k=0,\,k\geq 1$ in \eqref{string-general} and substituting it into the above equality one obtains \eqref{S1}. Equation \eqref{S1-x} is obtained by taking $x$-derivative of \eqref{S1}.
Equation \eqref{S3} can be derived from the defining equations \eqref{wave-L} and \eqref{wave-t}, which is a standard expression hence omitted. Taking the $x$-derivative in \eqref{S3} yields \eqref{S4}.
\end{proof}

We will apply in Appendix \ref{A-B} results of Propositions \ref{KS-p}, \ref{cor-string} to computation of higher Weil--Petersson volumes.

\section{Application to the Witten--Kontsevich tau-function}
\label{sectappl}
In this section we consider a particular solution of the KdV hierarchy \eqref{GD}. Namely, we are going to apply Theorem \ref{one-point-WK}, Theorem \ref{multi-point} and the Witten-Kontsevich theorem to derive $n$-point functions of the intersection numbers \eqref{def-gwpt}.  

To do so we first derive the {\it  initial datum} of the wave function $\psi^{^{WK}}\!\!(z;\bdt)$. Denote $f^{^{WK}}\!\!(z;x)=\psi^{^{WK}}\!\!(z;\bdt)|_{\bdt_{\geq 1}=\bdzero}.$ Noting that $u_{0}^{^{WK}}\!\!(x)=x,$ we have
\be
f_{xx}^{^{WK}}\!\!+2 \, x f^{^{WK}}\!\!=z^2f^{^{WK}}\!\!.
\ee
Solutions of this ODE can be expressed as linear combinations of  Airy functions. Imposing the asymptotic condition
 $f^{^{WK}}\!\!(z;x)=\le(1+\mathcal{O}(z^{-1})\ri)\exp(xz),$ it is easy to obtain
\be
\label{Airy}
f^{^{WK}}\!\!(z;x)=G(z)\sqrt{2\pi z}\,  {\rm e}^\frac{z^3}{3}2^\frac13  \Ai\le(2^{-\frac{2}{3}}(z^2-2\,x)\ri),
\ee
where $G(z)=1+\mathcal{O}(z^{-1})$ is yet to be determined and is a function independent of $x$.
The asymptotic expansion of \eqref{Airy} is of the indicated  form only in a suitable sector; the dual wave function admitting the required asymptotic in the same sector is found in Appendix \ref{Airyfun}.
To fix the gauge freedom we need to enforce the {\it string equation}
\be\label{string}
\sum_{k=0}^\infty t_{k+1} \frac{\p Z}{\p t_k}+\frac{t_1^2}{2}Z=\frac{\p Z}{\p t_1}.
\ee
which is derived by Witten \cite{Witten}.
One immediately reads off from \eqref{string} that
\be
\label{cval}
c_0=0, \, c_1=1,\, c_{k}=0\,(k\geq 2).
\ee
So, due to Definition \ref{def-KS}, the Kac--Schwarz operator associated to $Z(\bdt)$ reads
\be
S_z^{^{WK}}\!\!=\frac{1}{z}\p_z-\frac1{2z^2}-z.
\ee
Prop. \ref{cor-string} with the values \eqref{cval} yields immediately that 
\be
\label{eq1}
S_z^{^{WK}}\!\! f^{^{WK}}\!\!(z;x)  = -  \pa_x f^{^{WK}}\!\!(z;x) = \sqrt{2\pi z}\,  {\rm e}^\frac{z^3}{3}  2^\frac23 G(z) \Ai'\le(2^{-\frac{2}{3}}(z^2-2\,x)\ri).
\ee
On the other hand from \eqref{Airy} we have
\be
\label{eq2}
S_z^{^{WK}}\!\! f^{^{WK}}\!\!(z;x)  = \sqrt{2\pi z}\,  {\rm e}^\frac{z^3}{3}2^\frac13  \le( 2^\frac13 G(z) \Ai'\le(2^{-\frac{2}{3}}(z^2-2\,x)\ri) + \frac {G'(z)}z \Ai \le(2^{-\frac{2}{3}}(z^2-2\,x)\ri)\ri).
\ee
Comparing \eqref{eq1}, \eqref{eq2} we conclude that $G(z) $ is a constant and thus $G(z)\equiv 1$. 
\bd
We define $A^{^{WK}}\!\!(z):=\psi^{^{WK}}\!\!(z;\bdzero) = f^{^{WK}}\!\!(z;0), \ \ \ B^{^{WK}}\!\!(z):=\psi_x^{^{WK}}\!\!(z;\bdzero) = f_x^{^{WK}}\!\!(z;0).$
\ed
Setting $x=0$ in \eqref{eq2} yields
\be
A^{^{WK}}\!\!(z)=\sqrt{2\pi z}\,  {\rm e}^\frac{z^3}{3}2^\frac13  \Ai\le(2^{-\frac{2}{3}} z^2\ri) \simeq
c(z), \quad B^{^{WK}}\!\!(z)= -\sqrt{2\pi z}\,  {\rm e}^\frac{z^3}{3}2^\frac23 \, \Ai'\le(2^{-\frac{2}{3}} z^2\ri)\simeq  z\, q(z),\label{top-psi}
\ee
where the symbol $\simeq$ stands for asymptotic equivalence, and $c(z)$ and $q(z)$ are known as the Faber--Zagier series \cite{IZ,Buryak} which can be computed from the known expansions of Airy functions
\be
\Ai(z)\sim \frac{ {\rm e}^{-\zeta}}{2\sqrt{\pi} z^{1/4}} \sum_{k=0}^\infty \frac{ (6k-1)!!}{ (2k-1)!! } \frac{(-216\,\zeta)^{-k}}{k!}, \quad \zeta=\frac23 z^{3/2}, \quad z\to\infty, \quad |\arg z|<\pi.
\ee
So
\bea
&& c(z)=\sum_{k=0}^\infty C_k \,z^{-3k},\quad C_k=\frac{(-1)^k}{288^k}\frac{(6k)!}{(3k)!(2k)!},\label{def-c}\\
&& q(z)=\sum_{k=0}^\infty q_k \,z^{-3k},\quad q_k=\frac{1+6k}{1-6k} C_k. \label{def-q}
\eea
For future reference, we also have, from Prop. \ref{cor-string}
\begin{lemma} \label{lab}
$A^{^{WK}}\!\!(z),B^{^{WK}}\!\!(z)$ satisfy the following system of  ODEs 
\bea 
\le\{
\begin{array}{l}
 S_z^{^{WK}}\!\! A^{^{WK}}\!\!(z)=-B^{^{WK}}\!\!(z), \label{aabb}
\\[3pt] 
 S_z^{^{WK}}\!\! B^{^{WK}}\!\!(z)=-z^2 A^{^{WK}}\!\!(z). \label{bbaa}
\end{array}\ri .
\eea
\end{lemma}
The Faber--Zagier series satisfy the well-known relation
\be
c(z)\,q(-z)+c(-z)\,q(z)=2,
\ee
which expresses \eqref{wronsk} of Lemma \ref{wr}. Furthermore they possess the following properties.
\bl 
\label{2F3}  
The following equalities hold true:
\bea
&&c(z)\cdot c(-z)=\sum_{g=0}^\infty \frac{(6g-1)!!}{24^g\cdot g!} z^{-6g}, \\
&&q(z)\cdot q(-z)=-\sum_{g=0}^\infty\frac{6g+1}{6g-1} \frac{(6g-1)!!}{24^g\cdot g!} z^{-6g},\\
&&c(z)\cdot q(-z)=1-\frac12\sum_{g=1}^\infty \frac{(6g-5)!!}{24^{g-1}\cdot (g-1)!} z^{-6g+3},\\
&&q(z)\cdot c(-z)=1+\frac12\sum_{g=1}^\infty \frac{(6g-5)!!}{24^{g-1}\cdot (g-1)!} z^{-6g+3}.
\eea
\el
\begin{proof}
By straightforward computations and by using summation formul\ae\ for  hypergeometric series (i.e. Dixon's identities and contiguous identities \cite{DixonGen}).
\end{proof}

Due to the definition of $Z(\bdt)$, the relation between the intersection numbers and the $n$-point correlation functions of $Z(\bdt)$ consists simply in an evaluation at $\bdt=0$:
\be
\langle \tau_{k_1} \dots \tau_{k_n}\rangle= \langle \langle \tau_{k_1} \dots \tau_{k_n}\rangle\rangle|_{\bdt=0}.
\ee

\noindent {\it  Proof} of Theorem \ref{numbers}. \qquad 
Let us first consider one-point intersection numbers. According to Theorem \ref{one-point-WK} and the Witten-Kontsevich theorem, we have
\bea
 F_{1}^{^{WK}}\!\!(z)
&=&\frac14\le[q(z)\,c'(-z)+c'(z)\,q(-z)+\frac{c(-z)}{z}\,(q(z)+z\,q'(z))+\frac{c(z)}{z}\,(-q(-z)+z\,q'(-z))\ri]\nn\\
&=&z^2\left(1-\frac{c(z)c(-z)+q(z)q(-z)}{2}\right)= \sum_{g=1}^\infty  \frac{(6g-3)!!}{24^g \cdot g!} z^{-(6g-2)}.
\label{gen-0}
\eea
Thus we obtain
\be
\langle\tau_j\rangle=\left\{
 \begin{array}{lr}
\frac{1}{24^g \cdot g!} & j=3g-2,\\
 0 & \hbox{otherwise} , \end{array}\right.
 \label{one-p0}
\ee
i.e. $\langle \tau_{3g-2}\rangle=\langle \tau_{3g-2}\rangle_g=\frac{1}{24^g\cdot g!}.$ 

Now we consider the $n$-point function for intersection numbers of $\psi$-classes. Substituting \eqref{top-psi} into \eqref{TH2} and using Lemma \ref{2F3} we obtain
\bea
\Theta(z;0)&=&\frac{1}{2}\left(
\begin{array}{cc}
z\le(c(z)q(-z)-c(-z)q(z)\ri) & -2\,c(z)c(-z)\\
\\
-2z^2q(z)q(-z) &  -z\le(c(z)q(-z)-c(-z)q(z)\ri)\\
\end{array}
\right)\nn\\
&\nn\\
&=&\frac{1}{2}\left(
\begin{array}{cc}
-\sum_{g=1}^\infty \frac{(6g-5)!!}{24^{g-1}\cdot (g-1)!} z^{-6g+4} & -2 \sum_{g=0}^\infty \frac{(6g-1)!!}{24^g\cdot g!} z^{-6g}\\
\\
2 \sum_{g=0}^\infty \frac{6g+1}{6g-1} \frac{(6g-1)!!}{24^g\cdot g!} z^{-6g+2} &  \sum_{g=1}^\infty \frac{(6g-5)!!}{24^{g-1}\cdot (g-1)!} z^{-6g+4}\\
\end{array}
\right)=:M(z).\nn\\
\eea
So we have
\bea
F_{n}^{^{WK}}\!\!(z_1,\dots,z_n)=\frac{1}{n}\sum_{r\in S_{n}} P\le(r\ri) {\rm Tr} \le(M(z_{r_1})\cdots M(z_{r_{n}})\ri)-\delta_{n,2}\frac{z_1^2+z_2^2}{(z_1^2-z_2^2)^2}.
\eea
The theorem is proved.
\hfill
\QED
Theorem \ref{numbers} allows an efficient way of computing the intersection numbers of $\psi$-classes. 
E.g., for the two point function, write
\be
 {\rm Tr} \le(M(z_1) M(z_2)\ri)=\sum_{k_1,k_2=-1}^\infty a_{k_1,k_2} \, z_1^{-2k_1} z_2^{-2k_2},
\ee
where $a_{k_1,k_2}$ are rational numbers defined by
\be
a_{k_1,k_2}=\left\{
 \begin{array}{ll}
 \displaystyle 
 \frac{(6g_1-5)!!\cdot (6g_2-5)!!}{2\cdot 24^{g_1+g_2-2}\cdot (g_1-1)!\cdot (g_2-1)!}, & \hbox{if }~k_1=3g_1-2,\,k_2=3g_2-2, ~ g_1,g_2\geq 1,\\
 \\
  \displaystyle 
-\frac{(6g_1-1)!!\cdot (6g_2-1)!!}{24^{g_1+g_2}\cdot g_1! \cdot g_2!}  \frac{6g_2+1}{6g_2-1}, & \hbox{if }~ k_1=3g_1,\,k_2=3 g_2-1,~ g_1,g_2\geq 0,\\
\\
 \displaystyle 
-\frac{(6g_1-1)!!\cdot (6g_2-1)!!}{24^{g_1+g_2}\cdot g_1! \cdot g_2!}  \frac{6g_1+1}{6g_1-1}, & \hbox{if }~ k_1=3g_1-1,\,k_2=3 g_2,~ g_1,g_2\geq 0,\\
 0, & \hbox{otherwise}. \end{array}\right.
 \ee
Then formula \eqref{no-n} yields the following explicit expressions for two-point correlators
\bea
 \langle \tau_{3g_1+1}\tau_{3g_2+1}\rangle =\int\limits_{\overline{\mathcal M}_{g_1+g_2+1, 2}}\psi_1^{3g_1+1}\psi_2^{3g_2+1} &\& =\frac{\sum_{\ell=0}^{3g_1+1} (3g_1+2-\ell) \cdot a_{\ell-1,3g-\ell}}{(6g_1+3)!!\cdot (6g_2+3)!!},\quad g_1,g_2\geq 0,\\
 \langle \tau_{3g_1+2}\,\tau_{3g_2}\rangle  =\int\limits_{\overline{\mathcal M}_{g_1+g_2+1, 2}}\psi_1^{3g_1+2}\psi_2^{3g_2} &\& =\frac{\sum_{\ell=0}^{3g_1+2} (3g_1+3-\ell) \cdot a_{\ell-1,3g-\ell}}{(6g_1+5)!!\cdot (6g_2+1)!!},\quad g_1,g_2\geq 0.
\eea
For example, $ \int_{\ov{\mathcal M}_{51,2}} \psi_1^{61} \psi_2^{91}$ is equal to 
{\tiny $$\frac{9386050172836412587500989359024403743277403220016343379}{129591118281563315053010990258247407512356853373458520700907464141878255147238339379331556666041363562054992071393762192115131342143042355200000000000}.$$}
The above algorithm can also be easily applied to computation of high genus multipoint intersection numbers\footnote{A table of first few intersection numbers has been given in the appendix to the preprint version arXiv:1504.06452 of the present paper.}, e.g.
\bea
&& \langle \tau_{{20}}\tau_{{
21}}\tau_{{22}}\rangle ={\frac {59907930252114536543946157271}{
344102366437196621060106476460340816052999946240000}}\quad (\mbox{genus}\quad 21),\\
\nonumber\\
&&  \langle {\tau_{{8}}}^{2}{\tau_{{9}}}^{2}\rangle={\frac {15779395279487}{15064643317373337600}}\quad (\mbox{genus}\quad 11).
\eea

\paragraph{Relationship with topological recursion.}
Let us briefly comment on the relationship between the formul\ae\ \eqref{no-n} and functions of Eynard--Orantin type, sometimes referred to as solutions of the {\it topological recursions}. Recall that $Z(\bdt)$ admits the genus expansion
\be
\log Z(\bdt)=\sum_{g=0}^\infty \mathcal{F}_g(\bdt),
\ee
where $\mathcal{F}_g$ are genus $g$ free energies corresponding to the Witten--Kontsevich solution
\be
\mathcal{F}_g(\bdt)= \sum_{n=0}^\infty \frac{1}{n!} \sum_{k_1+\dots+k_n=3g-3+n}  \langle \tau_{k_1}\cdots\tau_{k_n}\rangle_{g,n} \,  t_{k_1}\cdots t_{k_n}.
\ee
Define
\be
W_{g,n}^{^{WK}}\!\!(z_1,\dots,z_n) = \nabla (z_1) \cdots \nabla (z_n)  \mathcal{F}_g(\bdt)\,|_{\bdt=\bdzero},
\ee
where we remind the reader that
\be
\nabla (z)=\sum_{k=0}^\infty \frac{(2k+1)!!}{z^{2k+2}} \frac{\p}{\p t_k}.
\ee
Hence we have
\be
F_n^{^{WK}}\!\!(z_1,\dots,z_n)=\sum_{g=0}^\infty W_{g,n}^{^{WK}}\!\!(z_1,\dots,z_n).
\ee
The functions $W_{g,n}^{^{WK}}\!\!(z_1,\dots,z_n)$ are introduced by Eynard and Orantin \cite{EO1}; see also \cite{Zhou1,Zhou2}. It is proved for example in \cite{Zhou2} that Eynard--Orantin's topological recursions for $W_{g,n}^{^{WK}}\!\!$ are equivalent to the Virasoro constraints \cite{DVV,DZ-norm}.

\section{Application to higher Weil--Petersson volumes}\label{wp}
In this section we apply Theorem \ref{multi-point} to the solution of higher Weil--Petersson volumes, for which we mean intersection numbers of mixed $\kappa$- and $\psi$-classes:
\be\label{go-kappa}
\langle \kappa_1^{d_1}\dots \kappa_\ell^{d_\ell} \tau_{k_1} \dots \tau_{k_n}\rangle_{g,n}:=\int_{\overline{\mathcal{M}}_{g,n}} \psi_1^{k_1}\dots\psi_n^{k_n}  \kappa_1^{d_1}\dots \kappa_\ell^{d_\ell}.
\ee
These numbers are zero unless
\be
\sum_{j=1}^\ell j \, d_j+\sum_{j=1}^n k_j=3g-3+n.
\ee

\begin{theorem} [\cite{MZ, KMZ,DZ-norm,MS,LX3}] \label{KMZ-thm}
The partition function $Z^{\kappa}(\bdt;\bds)$ defined by
\be
Z^{\kappa}(\bdt;\bds)=\exp\left(\sum_{g=0}^\infty \sum_{n,\ell\geq 0} \frac1{n!} \sum_{\sum jd_j+\sum k_j=3g-3+n}   \langle \kappa_1^{d_1}\dots \kappa_\ell^{d_\ell} \tau_{k_1} \dots \tau_{k_n}\rangle_{g,n}\,  t_{k_1}\cdots t_{k_n} \frac{\fs_1^{d_1}\dots \fs_\ell^{d_\ell}}{d_1! \cdots d_\ell!}\right)
\ee
is a particular tau-function of the KdV hierarchy \eqref{GD}. Moreover\footnote{
Recently Mattia Cafasso brought our  attention to an interesting paper \cite{ACM}, where shifts of arguments of tau-functions have been represented in the Grassmannian approach. It would be interesting to apply the methods of \cite{ACM} in order to get new insight to our computation of the $A(z;\bds),\,B(z;\bds)$ series (see below eqs. \eqref{azs}, \eqref{bzs}). We plan to do it in a subsequent publication.
},
\be\label{shifts}
Z^{\kappa}(\bdt;\bds)=Z(t_0,t_1,t_2-h_1(-\bds),t_3-h_2(-\bds),\dots, t_{k+1}-h_{k}(-\bds),\dots),
\ee
where $h_k(\bds)$ are polynomials in $\fs_1,\fs_2,\dots$ defined through
\be
\sum_{k=0}^\infty h_k(\bds) \, x^k=\exp\left(\sum_{j=1}^\infty {\fs_j} x^j \right).
\ee
\end{theorem}

\noindent Observe that, equivalently $Z^\kappa(\bdt;\bds)$ has the form
\be
Z^{\kappa}(\bdt;\bds)=\exp\left(\sum_{g=0}^\infty \sum_{n,q,\ell\geq 0} \sum_{{\sum m_j=n\atop 
\sum jd_j+\sum jm_j=3g-3+n}}  \!\!\!\!\!\!\!\!\!\!\! \langle \kappa_1^{d_1}\dots \kappa_\ell^{d_\ell} \tau_{0}^{m_0} \dots \tau_{q}^{m_q} \rangle_{g,n}\,  \frac{t_{0}^{m_0}\cdots t_{q}^{m_q}}{m_0!\cdots m_q!} \frac{\fs_1^{d_1}\dots \fs_\ell^{d_\ell}}{d_1! \cdots d_\ell!}\right).
\ee
It is easy to see that $Z(\bdt)=Z^{\kappa}(\bdt;\bdzero)$.

There exist several methods for computing the integrals \eqref{go-kappa}, including application of the Virasoro constraints \cite{M0,M,EO2,MS,Zhou2,LX3}, the {\it  quasi-triviality} approach \cite{DZ-norm,Zograf, IZ}, as well as an interesting method in the original paper of \cite{AC}.  
We propose a yet different approach,  based on Thm. \ref{multi-point} and Thm. \ref{KMZ-thm}. 

\subsection{Proof of Theorem \ref{WP-thm}}

Before entering into the proof proper, we need a few preparations.
\bl\label{kappa-def-lem}
The following formula holds true
\be
\exp\le(\sum_{m\geq 1} h_m(\bds) q_m\ri)=\sum_{\lambda\in\mathbb{Y}} \, \frac{\fs_\lambda}{m(\lambda)!}  \sum_{|\mu|=|\lambda|} L_{\lambda\mu} \, \frac{q_\mu}{m(\mu)!},
\ee
where $L_{\lambda\mu}$ is the transition matrix from the monomial basis to the power sum basis. We have used the shorthands 
\be
\fs_\l := \prod_{j=1}^{\ell(\l)} \fs_{\l_j} \ ,\ \ \ \  q_\mu := \prod_{j=1}^{\ell(\mu)} q_{\mu_j} \ .
\ee
\el
\begin{proof} By Taylor expansion of the exponential we have
\bea
\exp\le(\sum_{m\geq 1} h_m(\bds) q_m\ri)&\&=\sum_{n=0}^\infty \sum_{\sum k_j=n} \prod_{j=1}^\infty \frac{q_j^{k_j}}{k_j!} \prod_{j=1}^\infty h_j(\bds)^{k_j}\nn\\
&\&=\sum_{\mu\in\mathbb{Y}} \frac{q_\mu}{m(\mu)!} h_\mu(\bds)=\sum_{\mu\in\mathbb{Y}} \frac{q_\mu}{m(\mu)!} \sum_{|\lambda|=|\mu|} L_{\lambda\mu} \fs_\lambda,
\eea
where the last equality uses the definition of the transition matrix  between homogeneous basis and power sum basis.
\end{proof}
It is known that the matrix $(L_{\lambda\mu})$ equals to the product of the character table and the Kostka matrix. Now we are going to give two explicit formul\ae\ for $L_{\lambda\mu}$ for given $\lambda,\mu$. Write $\mathbb{Y}=\{\sigma_0\prec \sigma_1\prec \sigma_2\prec\dots\},$  where $\sigma_0=(0),\,\sigma_1=(1),\,\sigma_2=(2),\,\sigma_3=(1^2),$ etc. Let $p(n)$ be the number of partitions of weight $n$, i.e. $p(0)=p(1)=1,\,p(2)=2,\,p(3)=3,\dots;$ and let $r_n:=\sum_{k=1}^n p(k),\,n\geq 1,$ $r_0:=0.$   

According to the Macdonald's book \cite{Mac} $L_{\lambda\mu}$ counts the number of maps $f$ such that $f(\lambda)=\mu.$ 
The following lemma provides an efficient way for computing the matrices $L_{\lambda\mu}$.
\begin{lemma}
\label{Llm} For any $\lambda,\mu$ satisfying $|\lambda|=|\mu|,$ denote $v(\lambda)=\ell(\lambda)-m_1(\lambda)$. The following formula holds true:
\bea\label{kappa-2}
L_{\lambda\mu}  =   \sum_{1\leq  k_1,\dots,  k_{v(\lambda)} \leq \ell(\mu)}
\le(
{m_1(\lambda)\atop
\mu_1 - \sum_{j=1}^{ v(\lambda)} \delta_{1, k_j}  \lambda_j,\dots,\mu_{\ell(\mu)} - \sum_{j=1}^{ v(\lambda)}  \delta_{\ell(\mu), k_j}  \lambda_j}
\ri).
\eea 
\end{lemma}
\begin{proof}
By noticing that the combinatorial meaning of r.h.s. of \eqref{kappa-2} is the same as that of $L_{\lambda\mu}.$
\end{proof}

\br Let $\kappa_{\lambda\mu}=\frac{L_{\lambda\mu}}{m(\mu)!}$.  For any fixed weight $|\lambda|=|\mu|$, $(\kappa_{\lambda\mu})$ is a lower triangular matrix with integer entries which we call {\rm  $\kappa$-matrix}. Note that we have used the reverse lexicographic ordering for partitions. The first several $\kappa$-matrices are given by
\bea
&&(\kappa_{\lambda\mu})=(1),\quad |\lambda|=|\mu|=1;\qquad (\kappa_{\lambda\mu})=\le(\begin{array}{cc}
1 & 0\\
1 & 1\\
\end{array}\ri), \quad |\lambda|=|\mu|=2;\\
&&(\kappa_{\lambda\mu})=\left(
\begin{array}{ccc}
 1 & 0 & 0 \\
 1 & 1 & 0 \\
 1 & 3 & 1 \\
\end{array}
\right),\quad |\lambda|=|\mu|=3;\\
&&(\kappa_{\lambda\mu})=\left(
\begin{array}{ccccc}
 1 & 0 & 0 & 0 & 0 \\
 1 & 1 & 0 & 0 & 0 \\
 1 & 0 & 1 & 0 & 0 \\
 1 & 2 & 1 & 1 & 0 \\
 1 & 4 & 3 & 6 & 1 \\
\end{array}
\right),\quad |\lambda|=|\mu|=4;\\
&&(\kappa_{\lambda\mu})=\left(
\begin{array}{ccccccc}
 1 & 0 & 0 & 0 & 0 & 0 & 0 \\
 1 & 1 & 0 & 0 & 0 & 0 & 0 \\
 1 & 0 & 1 & 0 & 0 & 0 & 0 \\
 1 & 2 & 1 & 1 & 0 & 0 & 0 \\
 1 & 1 & 2 & 0 & 1 & 0 & 0 \\
 1 & 3 & 4 & 3 & 3 & 1 & 0 \\
 1 & 5 & 10 & 10 & 15 & 10 & 1 \\
\end{array}
\right),\quad |\lambda|=|\mu|=5.
\eea
We have the following observations:
\begin{itemize}
\item[(i)] The sum of the last row of $(\kappa_{\lambda\mu}),\,|\lambda|=|\mu|$ is a Bell number. Indeed taking in
\bea
\exp\le[\sum_{k\geq 1} h_k(\bds) q_k\ri] = \sum_{\lambda,\mu \in \mathbb Y}\frac{\kappa_{\lambda\mu}} {m(\lambda)!} \fs_\lambda q_\mu
\eea
$q_1=q_2=\dots=1$ we obtain
\be
\exp\le[
\exp\le[\sum_{j\geq 1 } \fs_k \ri]-1
\ri] =\sum_{\lambda \in \mathbb Y}\frac{\fs_\lambda}{m(\lambda)!}   \sum_{\mu \in \mathbb Y_{|\lambda|}}  \kappa_{\lambda\mu}.
\ee
Then taking $\fs_2=\fs_3=\dots=0$ we obtain
\be
\exp\le(\exp(\fs_1)-1\ri)   =\sum_{k=0}\frac{ \fs_1^k}{k!}   \sum_{\mu \in \mathbb Y_k}  \kappa_{(1^k)\mu},
\ee
i.e. $\sum_{\mu \in \mathbb Y_k}  \kappa_{(1^k)\mu}=B_k$ which is the $k$-{th} Bell number.
\item[(ii)] The sum of any row of $(\kappa_{\lambda\mu}),\,|\lambda|=|\mu|$ is a Bell number. More precisely,
\be\label{bell-id}
\sum_{\mu \in \mathbb Y_{|\lambda|}}  \kappa_{\lambda\mu}=B_{\ell(\lambda)}.
\ee
Indeed, due to the combinatorial meaning, for a fixed weight the sum only depends on the number of rows of $\lambda.$ Since in (ii) we have obtained the sum for $\lambda=(1^k),$ the formula \eqref{bell-id} follows for any $\lambda.$
\end{itemize}
\er

We are now in a position to prove Thm.\,\ref{WP-thm}. 

\noindent{\it Proof of Part I.}
Note that for any partition $\lambda=(\lambda_1\geq \lambda_2\geq \dots)\in\mathbb{Y}$,  we have
\be
\langle\kappa_{\lambda_1}\dots \kappa_{\lambda_{\ell(\lambda)}} \tau_{k_1}\dots \tau_{k_n}\rangle=\frac{\p^{n+\ell(\lambda)}\log Z^{\kappa}} {\p \fs_{\lambda_1}\dots \p \fs_{\lambda_{\ell(\lambda)}} \p t_{k_1}\dots \p t_{k_n}}(\bdzero;\bdzero).
\ee
According to Theorem \ref{KMZ-thm},
\bea
Z^{\kappa}(\bdt;\bds)= {\rm e}^{-\sum_{k\geq 1} h_k(-\bds)\, \p_{t_{k+1}}}Z(\bdt).
\eea
Therefore we can write the following expansion 
\bea
\log Z^{\kappa}(\bdt;\bds)=\sum_{\lambda\in\mathbb{Y}} \, \frac{(-1)^{\ell(\lambda)} \fs_\lambda}{m(\lambda)!}  \sum_{|\mu|=|\lambda|} L_{\lambda\mu} \, \frac{(-1)^{\ell(\mu)}}{m(\mu)!}\p_{t_{\mu_1+1}}\dots\p_{t_{\mu_{\ell(\mu)}+1}}\log Z(\bdt).
\label{ZK}
\eea
From \eqref{ZK} we collect the coefficient of $s_\lambda$ in the  generating function $F^{\kappa}_n(z_1,\dots, z_n; \bds)$. For this purpose we can use the generating formul\ae\ in Thm. \ref{one-point-WK} and Thm. \ref{multi-point} to arrive at 
\bea\label{PartII-prop}
&& \!\!\! \quad \sum_{k_1,\dots,k_n\geq 0} 
\langle \kappa_{\lambda_1}\dots\kappa_{\lambda_{\ell(\lambda)}} \tau_{k_1}\dots \tau_{k_n}\rangle \frac{(2k_1+1)!!}{z_1^{2k_1+2}}\cdots \frac{(2k_n+1)!!}{z_n^{2k_n+2}}\nn\\
&& \!\!\! =(-1)^{\ell(\lambda)} \sum_{|\mu|=|\lambda|}  \frac{L_{\lambda\mu}}{m(\mu)!} \res{w_1=\infty} \dots\res{w_{\ell(\mu)}=\infty} w^\mu \, F_{\ell(\mu)+n} (w_1,\dots,w_{\ell(\mu)}, z_1,\dots,z_n)\,d w_1 \cdots d w_{\ell(\mu)},
\eea
where for a partition $\mu,$ $w^\mu:=\frac{w_1^{2\mu_1+3}}{(2\mu_1+3)!!} \cdots \frac{w_{\ell(\mu)}^{2\mu_{\ell(\mu)}+3}}{(2\mu_{\ell(\mu)}+3)!!}.$  This concludes the proof of Part I.

\br Note that 
\bea
\p_{t_{k_1}} \dots \p_{t_{k_n}} \log Z^{\kappa}(\bdt;\bds)=  {\rm e}^{-\sum_{k\geq 1} h_k(-\bds)\, \p_{t_{k+1}}} \p_{t_{k_1}} \dots \p_{t_{k_n}}  \log Z(\bdt).
\eea
Let $F_n^{^{WK}}\!\!(z_1,\dots,z_n; \bdt)$ and $F^{\kappa}_n(z_1,\dots,z_n;\bdt;\bds)$ denote the generating functions of $n$-point correlation functions corresponding to $Z(\bdt)$ and $Z^\kappa(\bdt;\bds)$, respectively. Then we have
\bea
F^{\kappa}_n(z_1,\dots,z_n;\bdt;\bds)&=& {\rm e}^{-\sum_{k\geq 1} h_k(-\bds)\, \p_{t_{k+1}}}  F_n^{^{WK}}\!\!(z_1,\dots,z_n; \bdt)\nn\\
&=&\sum_{\lambda\in\mathbb{Y}} \, \frac{(-1)^{\ell(\lambda)} \fs_\lambda}{m(\lambda)!}  \sum_{|\mu|=|\lambda|} L_{\lambda\mu} \, \frac{(-1)^{\ell(\mu)}}{m(\mu)!}\p_{t_{\mu_1+1}}\dots\p_{t_{\mu_{\ell(\mu)}+1}}F_n^{^{WK}}\!\!(z_1,\dots,z_n; \bdt).
\eea
\er

\noindent{\it Proof of Part II.} By definition of the wave function we have
\bea
\psi^\kappa(z;\bdt;\bds)&=&\frac{Z^{\kappa}(\bdt-[z^{-1}];\bds)}{Z^{\kappa}(\bdt;\bds)}\exp\le(\vartheta(z;\bdt)\ri)\nn\\
&=&\frac{Z^{\kappa}(\bdt-[z^{-1}];\bds)}{Z^{\kappa}(\bdt;\bds)}\exp\le(\vartheta(z;\bdt)\ri) \exp\le(-\sum_{k=2}^\infty \frac{h_{k-1}(-\bds)z^{2k+1}}{(2k+1)!!}\ri) \exp\le(\sum_{k=2}^\infty \frac{h_{k-1}(-\bds)z^{2k+1}}{(2k+1)!!}\ri)\nn\\
&=&\exp\le(\sum_{k=2}^\infty \frac{h_{k-1}(-\bds)z^{2k+1}}{(2k+1)!!}\ri)
  \exp\le(-\sum_{k\geq 1} h_k(-\bds)\, \p_{t_{k+1}}\ri)  \psi^{^{WK}}\!\!(z;\bdt)\nn\\
 &=&\exp\le(\sum_{k=1}^\infty \frac{h_{k}(-\bds)z^{2k+3}}{(2k+3)!!}\ri) \sum_{\lambda\in\mathbb{Y}} \, \frac{(-1)^{\ell(\lambda)} \fs_\lambda}{m(\lambda)!}  \sum_{|\mu|=|\lambda|} L_{\lambda\mu} \, \frac{(-1)^{\ell(\mu)}}{m(\mu)!}\p_{t_{\mu_1+1}}\dots\p_{t_{\mu_{\ell(\mu)}+1}}\psi^{^{WK}}\!\!(z;\bdt).\nn
 \label{deformation-psi}\\
\eea
Similarly we have
\bea
\psi^\kappa_x(z;\bdt;\bds)=\exp\le(\sum_{k=1}^\infty \frac{h_{k}(-\bds)z^{2k+3}}{(2k+3)!!}\ri) \sum_{\lambda\in\mathbb{Y}} \, \frac{(-1)^{\ell(\lambda)} \fs_\lambda}{m(\lambda)!}  \sum_{|\mu|=|\lambda|} L_{\lambda\mu} \, \frac{(-1)^{\ell(\mu)}}{m(\mu)!}\p_{t_{\mu_1+1}}\dots\p_{t_{\mu_{\ell(\mu)}+1}}\psi_{x}^{^{WK}}\!\!(z;\bdt).
\eea

Let
\be
A(z;\bds):=\psi^\kappa(z;\bdzero;\bds),\quad B(z;\bds):=\psi^\kappa_x(z;\bdzero;\bds). \label{defAB}
\ee
Note that $Z^{\kappa}(\bdt;\bds)=1+\dots$ and that $Z^{\kappa}(\bdt;\bdzero)=Z(\bdt)$. Expanding in $s_1$, $s_2$, \dots ~ we have
\bea\label{azs}
A(z;\bds)=\sum_{k=0}^\infty \sum_{\lambda\in\mathbb{Y}_k} A^{\lambda}(z)\, \fs_{\lambda},\quad A^{(0)}(z)=A^{^{WK}}\!\!(z),\\
B(z;\bds)=\sum_{k=0}^\infty \sum_{\lambda\in\mathbb{Y}_k} B^{\lambda}(z)\, \fs_{\lambda},\quad B^{(0)}(z)=B^{^{WK}}\!\!(z)\label{bzs}
\eea
for some functions $A^\lambda(z)$ and $B^\lambda(z).$
The formul\ae\ \eqref{FSK1}, \eqref{FSK2} are then obtained from Thm.\,\ref{one-point-WK}, Thm.\,\ref{multi-point} and Thm.\,\ref{KMZ-thm}, and this proves Part II.  
The proof of Thm.\, \ref{WP-thm} is complete. \hfill $ \Box$

One can use the formul\ae\ in Lemma \ref{fundamental} and the formula \eqref{gen-Theta} recursively to obtain
\be
\nabla (z_1)\dots \nabla (z_n) \Psi.
\ee
At present we do not have a closed form for these generating functions of ``multi-point wave functions".

\subsection{Examples of computations of higher Weil--Petersson volumes}
In this section we give some explicit examples of calculations of higher Weil--Petersson volumes by using the main theorem and Thm.\,\ref{WP-thm}.
The following corollary of Lemma \ref{2F3} will be useful for these computations.
\bc \label{S(z)-zero} 
For the Witten--Kontsevich solution of the KdV hierarchy we have
\bea
&& \Rsol(z;\bdzero)=\sum_{g=0}^\infty \frac{(6g-1)!!}{24^g\cdot g!} z^{-6g},\quad \Rsol_x(z;\bdzero)= \sum_{g=1}^\infty \frac{(6g-5)!!}{24^{g-1}\cdot (g-1)!}z^{-6g+4},\\
&& \Rsol_{xx}(z;\bdzero)= \sum_{g=1}^\infty\frac{(6g-3)!!}{24^{g-1}\cdot (g-1)!} z^{-6g+2}.
\eea
\ec 
\begin{example} {\bf (Weil--Petersson volumes).}
Consider the special case $\bds = (s,0,0,\dots)$.
Then
\be
\log Z^{\kappa}(x,\bdzero;s,\bdzero)=
\sum_{g=0}^\infty \sum_{{n=0\atop 3g-3+n\geq 0}}^\infty \frac{s^{3g-3+n}}{(3g-3+n)! \, n!}  \langle \kappa_1^{3g-3+n}\tau_0^n\rangle_{g,n} \, x^n,
\ee
and by taking twice the $x$-derivatives we find
\be
u_0(x;s):=\p_x^2 \log Z^{\kappa}(x,\bdzero;s,\bdzero)=\sum_{n\geq 0} 
\frac{x^n}{n!} \sum_{{ g=0\atop  3g-1+ n\geq 0}}^\infty  \frac{ s^{3g-1+n}}{(3g-1+n)!} \langle \kappa_1^{3g-1+n} \tau_0^{n+2} \rangle_{g,n+2}.
\ee
The corresponding $n$-point  functions take the following form
\be
F^{^{WP}}_n(z_1,\dots,z_n;s):=\nabla (z_1) \cdots \nabla (z_n)  \log Z^{\kappa}(\bdt;s,\bdzero)\,|_{\bdt=0}=\sum_{g=0}^\infty W^{WP}_{g,n}(z_1,\dots,z_n;s),
\ee
where $W^{WP}_{g,n}(z_1,\dots,z_n;s)$ are rational functions of the ``Eynard--Orantin type":
\be
W_{g,n}^{WP}(z_1,\dots,z_n;s)=\sum_{d=0}^\infty \sum_{d+k_1+\cdots+k_n=3g-3+n} \langle \kappa_1^d \tau_{k_1} \dots \tau_{k_n}\rangle_g \frac{s^d}{d!} \prod_{i=1}^n \frac{(2k_i+1)!!}{z^{2k_i+2}}.
\ee
Define
\be
b_{n,g}= \left\{
 \begin{array}{ll}
 \displaystyle 
 \frac{s^{3g-1+n}}{(3g-1+n)!} \, \langle \kappa_1^{3g-1+n} \tau_0^{n+2} \rangle_{g,n+2},
 & \hbox{if } ~ 3g-1+n\geq 0,\\
 0, & \hbox{otherwise}. \end{array}\right.
\ee
An algorithm for computing $u_0(x;s)$ has been given in \cite{Zograf}, from a table of which we get
\bea
&& b_{0,0}=0,\quad b_{0,1}=\frac{s^2}{8\cdot 2!},\quad b_{0,2}=\frac{787s^5}{128\cdot 5!},\quad \dots,\\
&& b_{1,0}=1,\quad b_{1,1}=\frac{7 s^3}{6\cdot 3!},\quad b_{1,2}=\frac{1498069 s^6}{5760\cdot 6!},\quad \dots,\\
&& b_{2,0}=s,\quad b_{2,1}=\frac{529 s^4}{24 \cdot 4!},\quad b_{2,2}=\frac{10098059 s^7}{640\cdot 7!},\quad \dots,\\
&& b_{3,0}=5\frac{s^2}{2!},\quad b_{3,1}=\frac{16751s^5}{24\cdot 5!},\quad b_{3,2}=\frac{7473953867 s^8}{5760\cdot 8!}, \quad \dots.
\eea
Substituting the above initial data into \eqref{Rsolv-gen},\eqref{n-point-gen} we obtain
\be
&&W^{WP}_{0,2}=0,\quad W^{WP}_{1,2}=\frac{s^2}{16 z_1^2 z_2^2}+s\le(\frac{1}{4z_1^4 z_2^2}+\frac{1}{4z_2^4 z_1^2}\ri)+\frac{3}{8z_1^4 z_2^4}+\frac{5}{8z_1^6 z_2^2}+\frac{5}{8 z_1^2 z_2^6},\\
&&W^{WP}_{2,2}=\frac{787 s^5}{15000 z_1^2 z_2^2}+s^4 \left(\frac{1085}{4608 z_1^4 z_2^2}+\frac{1085}{4608 z_1^2 z_2^4}\right)+s^3 \left(\frac{551}{576 z_1^6 z_2^2}+\frac{7}{8 z_1^4 z_2^4}+\frac{551}{576 z_1^2 z_2^6}\right)\nn\\
&&\qquad\quad+s^2 \left(\frac{399}{128 z_1^8 z_2^2}+\frac{181}{64 z_1^6 z_2^4}+\frac{181}{64 z_1^4 z_2^6}+\frac{399}{128 z_1^2 z_2^8}\right)\nn\\
&&\qquad\quad+s \left(\frac{231}{32 z_1^{10} z_2^2}+\frac{203}{32 z_1^8 z_2^4}+\frac{105}{16 z_1^6 z_2^6}+\frac{203}{32 z_1^4 z_2^8}+\frac{231}{32 z_1^2 z_2^{10}}\right)+\frac{1155}{128 z_1^{12} z_2^2}+\frac{945}{128 z_1^{10} z_2^4}\nn\\
&&\qquad\quad+\frac{1015}{128 z_1^8 z_2^6}+\frac{1015}{128 z_1^6 z_2^8}+\frac{945}{128 z_1^4 z_2^{10}}+\frac{1155}{128 z_1^2 z_2^{12}},\\
&&W^{WP}_{0,3}=\frac{1}{z_1^2 z_2^2 z_3^2},\\
&&W^{WP}_{1,3}=\frac{7 s^3}{36 z_1^2 z_2^2 z_3^2}+s^2 \left(\frac{13}{16 z_1^4 z_2^2 z_3^2}+\frac{13}{16 z_1^2 z_2^4 z_3^2}+\frac{13}{16 z_1^2 z_2^2 z_3^4}\right)\nn\\
&&\qquad \quad +s \left(\frac{5}{2 z_1^6 z_2^2 z_3^2}+\frac{9}{4 z_1^4 z_2^4 z_3^2}+\frac{9}{4 z_1^4 z_2^2 z_3^4}+\frac{5}{2 z_1^2 z_2^6 z_3^2}+\frac{9}{4 z_1^2 z_2^4 z_3^4}+\frac{5}{2 z_1^2 z_2^2 z_3^6}\right)\nn\\
&&\qquad\quad +\frac{35}{8 z_1^8 z_2^2 z_3^2}+\frac{15}{4 z_1^6 z_2^4 z_3^2}+\frac{15}{4 z_1^6 z_2^2 z_3^4}+\frac{15}{4 z_1^4 z_2^6 z_3^2}+\frac{9}{4 z_1^4 z_2^4 z_3^4}+\frac{15}{4 z_1^4 z_2^2 z_3^6}\nn\\
&&\qquad\quad+\frac{35}{8 z_1^2 z_2^8 z_3^2}+\frac{15}{4 z_1^2 z_2^6 z_3^4}+\frac{15}{4 z_1^2 z_2^4 z_3^6}+\frac{35}{8 z_1^2 z_2^2 z_3^8}.
\ee
The expressions for $W^{WP}_{1,2},W^{WP}_{0,3}$ coincide with those derived in \cite{Zhou2}; the function $W^{WP}_{0,2}$ in \cite{Zhou2} is not vanishing since some natural supplementary definitions for intersection numbers are used in \cite{Zhou2}. Let $Vol_{g,n}(L_1,\dots,L_n):=Vol(\mathcal{M}_{g,n}(L_1,\dots,L_n))$ denote the Weil--Petersson volumes and let
\be
v_{g,n}(L_1,\dots,L_n):=\frac{Vol_{g,n}(L_1,\dots,L_n)}{(2\pi)^{3g-3+n}}.
\ee
It was shown in \cite{M0,M} that
\be
v_{g,n}(L_1,\dots,L_n)=\sum_{d+k_1+\cdots k_n=3g-3+n} \frac{\langle\kappa_1^{d} \tau_{k_1} \dots \tau_{k_n} \rangle_g}{d!\, k_1! \cdots \, k_n!} \, L_1^{2k_1}\cdots L_n^{2k_n}.
\ee
We remark that the relationship between $v_{g,n}$ and $W^{WP}_{g,n}$ is a Laplace transform \cite{MS,EO2,Zhou1}:
\be
W^{WP}_{g,n}(z_1,\dots,z_n;s=1)=2^n \int_0^\infty \cdots \int_0^\infty L_1\cdots L_n \cdot v_{g,n}(L_1,\dots,L_n)\cdot  {\rm e}^{-\sum_{i=1}^n \sqrt{2} z_i L_i}\,d L_1\dots d L_n.
\ee
\end{example}

{Note that in the above example we have not used Thm. \ref{WP-thm}. More generally, given any initial value problem of the KdV hierarchy with an initial data $u(x,\bdzero)=u_0(x)\in\mathbb{C}[[x]]$ one can compute the resolvent function ${\mathcal R(z; x)}$ of the Lax operator by solving the ODE \eqref{ODE-S}. Then, with the help of the Main Theorem, the generating functions of multipoint correlators readily follow.
}

Let us now give examples of application of Thm. \ref{WP-thm}.

\begin{example} {\bf  (linear insertion of $\kappa$-classes).}
For $n=0,\,\lambda=(j)$ we have
\bea
&& \quad \langle\kappa_j\rangle=\res{w=\infty} \frac{-w^{2j+3}}{(2j+3)!!} F_1(w)\,dw
 =\res{w=\infty} \frac{-w^{2j+3}}{(2j+3)!!} \sum_{g=1}^\infty \frac{(6g-3)!!}{24^g g!} w^{-(6g-2)}dw.
\eea
So we have
\be
\langle\kappa_{3g-3}\rangle_{g,0}=\frac{1}{24^g\cdot g!}.
\ee
For $n=1,\,\lambda=(j)$ we have
\bea
&& ~\quad\sum_{k_1\geq 0} \langle\kappa_j \tau_{k_1}\rangle \frac{(2k_1+1)!!}{z^{2k_1+2}}=\res{w=\infty} \frac{-w^{2j+3}}{(2j+3)!!} F_2(w,z)\,dw\\
&& =\res{w=\infty} \frac{-w^{2j+3}}{(2j+3)!!} \left( \frac{{\rm Tr} \le( M(w) M(z)\ri)}{(w^2-z^2)^2} -\frac{w^2+z^2}{(w^2-z^2)^2}\right)dw\\
&& = \frac{1}{(2j+3)!!}{\rm Tr} \le( \frac{1}{2z} \le[\p_{z}\le(z^{2j+2} M(z)\ri)\ri]_+\cdot M(z)\ri)  -  \frac{1}{(2j+1)!!}z^{2j+2}.
\eea
Here $``+"$ means taking the polynomial part in the $z$ expansion at infinity. Particularly, if $n=1,\,j=1$ then we have
\be
\frac{1}{2z} \le[\p_{z}\le(z^{4} M(z)\ri)\ri]_+=\le(\begin{array}{cc}
-\frac12 & -2 z^2\\
-3z^4 & \frac12
\end{array}\ri);
\ee
hence
\bea
&& \quad \sum_{k\geq 0} \langle\kappa_1 \tau_{k}\rangle \frac{(2k+1)!!}{z^{2k+2}}=\frac1{5!!}{\rm Tr} \le( \frac{1}{2z} \le[\p_{z}\le(z^{4} M(z)\ri)\ri]_+\cdot M(z)\ri)-\frac1{3!!}z^4\nn\\
&&=\frac{1}{5!!}\sum_{g=1}^\infty \left(\frac12 \frac{(6g-5)!!}{24^{g-1}\cdot (g-1)!} -2\frac{6g+1}{6g-1} \frac{(6g-1)!!}{24^g\cdot g!} +3 \frac{(6g-1)!!}{24^g\cdot g!}\ri) z^{-6g+4}\nn\\
&&=3\sum_{g=1}^\infty \frac{(12g^2-12g+5)(6g-5)!!}{5!! \cdot 24^g\cdot g!}z^{-6g+4}.
\eea
This means that
\be
\langle\kappa_1 \tau_{3g-3}\rangle_{g,1}=3\frac{12g^2-12g+5}{5!! \cdot 24^g\cdot g!},\, g\geq 1.
\ee If $n=1,\,j=2$ then we have
\be
\frac{1}{2z} \le[\p_{z}\le(z^{6} M(z)\ri)\ri]_+=\left(
\begin{array}{cc}
 -z^2 & -3 z^4 \\
 \frac{7}{8}-4 z^6 & z^2 \\
\end{array}
\right);
\ee
whence a straightforward manipulation of series shows
\bea
&& \quad \sum_{k\geq 0} \langle\kappa_2 \tau_{k}\rangle \frac{(2k+1)!!}{z^{2k+2}}=\frac{1}{7!!} {\rm Tr} \le( \frac{1}{2z} \le[\p_{z}\le(z^{6} M(z)\ri)\ri]_+\cdot M(z)\ri)-\frac{1}{5!!} z^6\nn\\
&&=3\sum_{g=1}^\infty \frac{(72g^3-132g^2+95g-35)(6g-7)!!}{7!!\cdot 24^g\cdot g!}z^{-6g+6}.
\eea
This means that
\be
\langle\kappa_2 \tau_{3g-4}\rangle_{g,1}=3\frac{72g^3-132g^2+95g-35}{7!!\cdot 24^g\cdot g!},\qquad g\geq 2.
\ee
Similar  computations lead to 
\bea
\langle\kappa_3 \tau_{3g-5}\rangle_{g,1}=\frac{1296 g^4 - 3888 g^3+ 4482 g^2- 2835 g+945}{9!! \cdot 24^g\cdot g!},\qquad g\geq 2.
\eea
In general, we have
\be
\langle\kappa_j\tau_{3g-j-2}\rangle_{g,1}\sim \frac{6^{j+1} g^{j+1}}{(2j+3)!!\cdot 24^g\cdot g!},\quad g\rightarrow \infty.
\ee
For $n\geq 2,\,\lambda=(j)$ we have
\bea
&&\sum_{k_1,\dots,k_n\geq 0} \langle\kappa_j \tau_{k_1}\dots \tau_{k_n}\rangle \frac{(2k_1+1)!!}{z_1^{2k_1+2}}\cdots \frac{(2k_n+1)!!}{z_n^{2k_n+2}}=\res{w=\infty} \frac{-w^{2j+3}}{(2j+3)!!} F_{n+1}(w,z_1,\dots,z_n)\,dw\\
&&=\res{w=\infty} \frac{w^{2j+3}}{(2j+3)!!} \sum_{r\in S_{n}}  \frac{{\rm Tr} \le( M(w) M(z_{r_1})\cdots M(z_{r_{n}})\ri)}{(w^2-z_{r_1}^2)(z_{r_n}^2-w^2)
\prod_{j=1}^{n-1}(z_{r_j}^2-z_{r_{j+1}}^2)}\\
&&=\frac1{(2j+3)!!}\res{w=\infty} \sum_{r\in S_{n}}  \frac{{\rm Tr} \le( M(w) M(z_{r_1})\cdots M(z_{r_{n}})\ri)}{\prod_{j=1}^{n}(z_{r_j}^2-z_{r_{j+1}}^2)}\le(\frac{w^{2j+3}}{w^2-z_{r_1}^2}-\frac{w^{2j+3}}{w^2-z_{r_n}^2}\ri)  \,dw \\
&&=-\frac1{(2j+3)!!} \sum_{r\in S_{n}}  \frac{{\rm Tr} \le( \bigl((z_{r_1}^{2j+2} M(z_{r_1}))_+-(z_n^{2j+2} M(z_{r_n}))_+\bigr) M(z_{r_1})\cdots M(z_{r_{n}})\ri)}{\prod_{j=1}^{n}(z_{r_j}^2-z_{r_{j+1}}^2)}.
\eea
So we can also collect the following generating function for linear insertion of $\kappa$-classes:
\be
\sum_{j\geq 1} \sum_{k_1,\dots,k_n\geq 0} \langle\kappa_j \tau_{k_1}\dots \tau_{k_n}\rangle \frac{(2j+3)!!}{w^{2j}}\frac{(2k_1+1)!!}{z_1^{2k_1+2}}\cdots \frac{(2k_n+1)!!}{z_n^{2k_n+2}}=\left[ w^4 F_{n+1}(w,z_1,\dots,z_n)\right]_-,
\ee
where $``-"$ means taking the negative part in the $w$ expansion at $\infty.$
\end{example}

\begin{example} [{\bf linear deformation of the wave function}] For $j\geq 1,$ we have
\bea
A^{(j)}(z)&=&-\frac{z^{2j+3}}{(2j+3)!!}A^{(0)}(z)-\sum_{|\mu|=j} L_{(j),\mu} \, \frac{(-1)^{\ell(\mu)}}{m(\mu)!}\p_{t_{\mu_1+1}}\dots\p_{t_{\mu_{\ell(\mu)}+1}}\psi^{^{WK}}\!\!(z;\bdt)|_{\bdt=\bdzero}\nn\\
&=&-\frac{z^{2j+3}}{(2j+3)!!}c(z)+\psi_{t_{j+1}}^{^{WK}}\!(z;\bdzero).
\eea
By using Lemma \ref{fundamental} and Corollary \ref{S(z)-zero} we obtain
\bea
&&\nabla (w) \psi^{^{WK}}\!\!(z;\bdt)|_{\bdt=\bdzero}\nn\\
&&= \frac{2S(w;\bdzero)B^{^{WK}}\!\!(z)-\Rsol_x(w;\bdzero) A^{^{WK}}\!\!(z)}{2\,(w^2-z^2)}\nn\\
&&=\frac{2 \sum_{g=0}^\infty \frac{(6g-1)!!}{24^g\cdot g!} w^{-6g} z\, q(z)- \sum_{g=1}^\infty \frac{(6g-5)!!}{24^{g-1}\cdot (g-1)!}w^{-6g+4} c(z)}{2\,(w^2-z^2)}\nn\\
&&=\frac{z \, q}{w^2}+\frac{2z^3 q-c}{2\,w^4}+\frac{4z^5 q-2 z^2 c}{4\,w^6}+\frac{(8z^7 +5 z)\, q-4 z^4 c}{8\,w^8}
+\frac{(16 z^9 +10 z^3) q-(8 z^6 +35) c}{16\, w^{10}}+\mathcal{O}(w^{-12}).\nn\\
\eea
So we have
\bea
&& A^{(1)}(z)=-\frac{z^{5}}{5!!}c+\frac{z^5}{5!!} q-\frac{z^2}{2\cdot 5!!}  c=-\frac1{24}z^{-1}+\frac{77}{576}z^{-4}+\mathcal{O}(z^{-7}),\\
&& A^{(2)}(z)=-\frac{z^{7}}{7!!}c+\frac{8z^7+5z}{8\cdot 7!!} q-\frac{z^4}{2\cdot 7!!}  c=\frac{1}{48}z^{-2}-\frac{13}{144}z^{-5}+\mathcal{O}(z^{-8}),\\
&& A^{(3)}(z)=-\frac{z^9}{9!!}c+\frac{16 z^9 +10 z^3}{16 \cdot 9!!} q-\frac{8 z^6 +35}{16 \cdot 9!!} c=-\frac{11}{1152}z^{-3}+\frac{1639}{27648}z^{-6}+\mathcal{O}(z^{-9}),
\eea
etc. In the above formul\ae\, it is understood that $q=q(z),\,c=c(z)$ are Faber--Zagier series. {As it was expected, indeed,  all deformation terms $A^{\lambda}(z),\,\lambda\neq (0)$ contain only negative powers in their expansions at $z\to\infty$.} Similarly, we have
\bea
&&B^{(j)}(z)=-\frac{z^{2j+4}}{(2j+3)!!} q(z)+\psi_{t_0 t_{j+1}}^{^{WK}}\!\!(z;\bdzero),\\
&&\nabla (w) \psi_x^{^{WK}}\!\!(z;\bdt)|_{\bdt=\bdzero}=\frac{\Rsol_x(w;\bdzero)B^{^{WK}}\!\!(z)-\le[\Rsol_{xx}(w;\bdzero)-2S(w;\bdzero)z^2\ri] A^{^{WK}}\!\!(z)}{2\,(w^2-z^2)}\nn\\
&&=\frac{\sum_{g=1}^\infty \frac{(6g-5)!!}{24^{g-1}\cdot (g-1)!}w^{-6g+4} z \,q(z)-\le[ \sum_{g=1}^\infty\frac{(6g-3)!!}{24^{g-1}\cdot (g-1)!} w^{-6g+2}-2 \sum_{g=0}^\infty \frac{(6g-1)!!}{24^g\cdot g!} w^{-6g} z^2\ri] c(z)}{2\,(w^2-z^2)}\nn\\
&&=\frac{z^2 c}{w^2}+\frac{2 z^4 c+ z q}{2w^4}+\frac{(4z^6 -6) c+2z^3 q}{4 w^6}+\frac{(8z^8 -7 z^2) c+4 z^5 q}{8 w^8}+\mathcal{O}(w^{-10}).
\eea
And we have
\bea
&& B^{(1)}(z)=-\frac{z^{6}}{5!!}q+\frac{z^3}{2\cdot 5!!} q + \frac{4z^6-6}{4\cdot 5!!}  c=-\frac1{24}+\frac{79}{576}z^{-3}+\frac{18095}{27648}z^{-6}+\mathcal{O}(z^{-9}),\\
&& B^{(2)}(z)=-\frac{z^{8}}{7!!}q+\frac{z^5}{2\cdot 7!!} q + \frac{8z^8-7z^2}{8\cdot 7!!}  c=-\frac{1}{48}z^{-1}-\frac{55}{576}z^{-4}-\frac{31603}{55296}z^{-7}+\mathcal{O}(z^{-10}).
\eea
These expressions agree with expression derived from a less straightforward method in Appendix \ref{A-B}. 
\end{example}
\begin{example} [{\bf higher deformation of the wave function}]
Consider deformation of the wave function associated to partitions of form $(1^k),\,k\geq 1.$ For $k=1$, it has been solved above.  Let $k=2;$ we have
\bea
&&A^{(1^2)}(z)=\le(\frac{z^7}{2\cdot 7!!}+\frac{z^{10}}{2 \cdot 5!!^2}\ri)c-\frac{z^5}{5!!} \psi_{t_2}^{^{WK}}\!\!(z;\bdzero)+
\frac{1}{2!}\sum_{|\mu|=2} L_{(1^2),\mu} \frac{(-1)^{\ell(\mu)}}{m(\mu)!}\p_{t_{\mu_1+1}}\dots\p_{t_{\mu_{\ell(\mu)}+1}}\psi^{^{WK}}\!\!(z;\bdt)|_{\bdt=\bdzero}\nn\\
&&=\le(\frac{z^7}{2\cdot 7!!}+\frac{z^{10}}{2 \cdot 5!!^2}\ri)c - \frac{z^5}{5!!} \frac{4z^5 q-2 z^2 c}{4\cdot 5!!}+
\frac{1}{2!}\le(-\frac{(8z^7 +5 z)\, q-4 z^4 c}{8\cdot 7!!}+ \psi_{t_2t_2}^{^{WK}}\! (z;\bdzero)\ri)
\eea
and
\bea
&&B^{(1^2)}(z)=\le(\frac{z^8}{2\cdot 7!!}+\frac{z^{11}}{2 \cdot 5!!^2}\ri)q -\frac{z^5}{5!!} \psi_{t_0 t_2}^{^{WK}}\!(z;\bdzero)+
\frac{1}{2!}\sum_{|\mu|=2} L_{(1^2),\mu} \frac{(-1)^{\ell(\mu)}}{m(\mu)!}\p_{t_{\mu_1+1}}\dots\p_{t_{\mu_{\ell(\mu)}+1}}\psi_x^{^{WK}}\!(z;\bdt)|_{\bdt=\bdzero}\nn\\
&&=\le(\frac{z^8}{2\cdot 7!!}+\frac{z^{11}}{2 \cdot 5!!^2}\ri)q - \frac{z^5}{5!!} \frac{(4z^6 -6) c+2z^3 q}{4\cdot 5!!}+
\frac{1}{2!}\le(-\frac{(8z^8 -7 z^2) c+4 z^5 q}{8\cdot 7!!} + \psi_{t_0 t_2 t_2}^{^{WK}}\!\! (z;\bdzero)\ri). 
\eea
By using \eqref{1-bbt} one can derive that
\be
\psi_{t_2t_2}^{^{WK}}\!(z;\bdzero)=\frac{4z^{10}-5z^4}{4\cdot 5!!^2}c+\frac{z}{4 \cdot 5!!} q, \quad \psi_{t_0 t_2 t_2}^{^{WK}}\!(z;\bdzero)=\frac{z^{11}q}{5!!^2}-\frac{z^5 q}{180}-\frac{z^2 c}{60}.
\ee
So we have
\bea
A^{(1^2)}(z)&=&\left(\frac{z^{10}}{225}+\frac{11 z^7}{1575}-\frac{ z^4}{2520}\right)c+\le(-\frac{z^{10}}{225}-\frac{ z^7}{210}+\frac{3 z}{560}\ri)q\nn\\
&=&\frac{37}{1152}z^{-2}-\frac{28249}{138240}z^{-5}+\mathcal{O}(z^{-8}),\\
 B^{(1^2)}(z)&=&\le(-\frac{z^{11}}{225}-\frac{z^8}{210}+\frac{z^5}{150}-\frac{z^2}{240}\ri) c+\left(\frac{ z^{11}}{225}+\frac{4  z^8}{1575}-\frac{13  z^5}{2520}\right)q\nn\\
&=&-\frac{35}{1152}z^{-1}+\frac{29051}{138240}z^{-4}+\mathcal{O}(z^{-7}).
\eea
Using Part I of Theorem \ref{WP-thm} we find
\be
\langle \kappa_1^2 \tau_2 \rangle_{2,1}=\frac{139}{11520},\quad \langle \kappa_1^2 \tau_5 \rangle_{3,1}=\frac{3781}{2903040}, \quad  \langle \kappa_1^2 \tau_8 \rangle_{4,1}=\frac{48689}{928972800}
\ee
and
\bea
F^{\kappa}_2(z,w;\fs_1=s,\bdzero)&=&s^2 \left(\frac{300825}{1024 w^8 z^8}+\frac{399}{128 w^8 z^2}+\frac{181}{64 w^6 z^4}+\frac{181}{64 w^4 z^6}+\frac{399}{128 w^2 z^8}+\frac{1}{16 w^2 z^2}+\dots\right)\nn\\
&&+s \left(\frac{231}{32 w^{10} z^2}+\frac{203}{32 w^8 z^4}+\frac{105}{16 w^6 z^6}+\frac{203}{32 w^4 z^8}+\frac{1}{4 w^4 z^2}+\frac{231}{32 w^2 z^{10}}+\frac{1}{4 w^2 z^4}+\dots\right)\nn\\
&&+\le(\frac{1015}{128 w^8 z^6}+\frac{1015}{128 w^6 z^8}+\frac{5}{8 w^6 z^2}+\frac{3}{8 w^4 z^4}+\frac{5}{8 w^2 z^6}+\dots\ri)+o(s^2).
\eea
\end{example}
{An alternative method for computing $A(z;\bds),B(z;\bds)$ is presented in Appendix \ref{A-B}.}
\section{Further remarks}
\label{furtherrems}
It would be interesting to prove directly the equivalence between the formula \eqref{no-n},  the (explicit) integral/recursive formul\ae\  of ``$n$-point functions" given by Okounkov \cite{Ok}, Liu--Xu \cite{LX1,LX2}, Br\'ezin--Hikami \cite{BH1,BH2}, and Kontsevich's main identity \cite{Kontsevich}. We indicate the relation between these three as follows. Recall that, given a formal (divergent) series of the form 
\be
f(z)=\sum_{k=0}^\infty v_k \frac{(2k+1)!!}{z^{2k+2}},
\ee
it can be re-summed by a suitable version of the Borel summation method. Noting that
\be
(2k+1)!!=\frac{2^{k+1}}{\sqrt{\pi}} \int_0^\infty u^{\frac{2k+1}2}  {\rm e}^{-u}du,\quad k\geq 0
\ee
we have, integrating term-by-term 
\be
\sum_{k=0}^\infty v_k \frac{(2k+1)!!}{z^{2k+2}}=\frac{2\,z}{\sqrt{\pi}} \int_0^\infty s^{\frac{1}{2}}\,  \hat{f}(s) \,  {\rm e}^{-sz^2}  \,ds,
\ee
where $\hat{f}(s):=\sum_{k=0}^\infty v_k (2s)^k$ is the Borel re-summation of $f(z)$.
In general we say that $f(z)$ is Borel summable if the function $\hat{f}(s)$ (the ``Borel transform'' of $f$)  has a nonzero radius of convergence, and it can be extended analytically along a strip surrounding the real positive axis. 
It then follows that the relation between the generating functions of intersection numbers that we have constructed in this paper and Okounkov's  generating functions is precisely that the latter are Borel transforms of the former, or, which is the same, that the former are asymptotic expansions of the Laplace transforms of the latter. The relation between our generating function and Kontsevich's is simpler.
To be precise, let us introduce the following notations
\bea
F^{^{OK}}_n(x_1,\dots,x_n)&=&\sum_{k_1,\dots,k_n=0}^\infty  {  \langle \tau_{k_1} \dots \tau_{k_n}\rangle} \, {x_1^{k_1}\cdots x_n^{k_n}},\\
F^{^{K}}_n(z_1,\dots,z_n)&=& \sum_{k_1,\dots,k_n=0}^\infty  \frac{(2k_1-1)!!}{z_1^{2k_1+2}}\cdots \frac{(2k_n-1)!!}{z_n^{2k_n+2}}\,  \langle \tau_{k_1} \dots \tau_{k_n}\rangle.
\eea
Then we have
\bea\label{Laplace}
\!\!\!\!\!\! F_n^{^{WK}}\!\!(z_1,\dots,z_n) \! & \! =\! &\! \frac{2^n\,z_1\cdots z_n}{\pi^{\frac n2}} \int_0^\infty\cdots\int_0^\infty (x_1\cdots x_n)^{\frac{1}{2}}\,  F^{^{OK}}_n(2x_1,\dots,2x_n) \,  {\rm e}^{-x_1z_1^2-\dots -x_n z_n^2}  \,dx_1\cdots dx_n,\\
\!\!\!\!\!\! F_n^{^{WK}}\!\!(z_1,\dots,z_n)\! & \! =\! &\! (-1)^n 
\le(\prod_{j=1}^n \frac{\p}{\p z_j} z_j\ri)
 \le( F^K_n(z_1,\dots,z_n)\ri).
\ee

For $n=1$, one can directly verify \eqref{Laplace} (actually in the case $n=1$ the generating series of Okounkov and ours are both well-known). Indeed, 
\be
F^{^{OK}}_1(2x_1)=\frac{1}{4x_1^2}\big( {\rm e}^{\frac{x_1^3}{3}}-1\big)=\sum_{g=1}^\infty \frac{1}{24^g g!} (2x_1)^{3g-2}.
\ee
Recall that
\be
F_{1}^{^{WK}}\!\!(z_1)=\sum_{g=1}^\infty \frac{(6g-3)!!}{24^g g!} z_1^{-(6g-2)}.
\ee
So it is straightforward to verify \eqref{Laplace} for $n=1.$ However, it appears that for $n\geq 2$ a direct verification of \eqref{Laplace} is not trivial.

The final remark is that the formula \eqref{n-point-gen} for multi-point correlation functions possesses certain universality in tau-symmetric integrable systems. Indeed, it can be generalized to the Gelfand-Dickey hierarchy and more generally to the Drinfeld-Sokolov hierarchy, as well as to the Jimbo--Miwa--Ueno isomonodromic problems \cite{B, JMU}, which will be presented in a separate publication \cite{BDY1}. It would be interesting to investigate whether the formula works for all integrable hierarchies of topological (or cohomological) type associated to semisimple Frobenius manifolds \cite{DZ-norm, DLYZ}: the first nontrivial examples in this investigation would be the intermediate long wave hierarchy \cite{Buryak2}, the discrete KdV hierarchy \cite{DLYZ} and the extended Toda hierarchy \cite{CDZ}.


\appendix
\section{Some useful formul\ae\ }
\label{Airyfun}
The wave functions associated to the Witten--Kontsevich tau-function $Z(\bdt)$ at $t_1=t_2=\dots=0$ satisfy
\bea
&&\psi(z;t_0,0,0,\dots) \, = \, \sqrt{2\pi z} \, {\rm e}^\frac{z^3}{3} \, 2^\frac13  \, \Ai\le(\xi\ri),\\
&&\psi^*(z;t_0,0,0,\dots) \, = \, {\rm e}^{\frac{\pi i}{6}}\sqrt{2\pi z} \, {\rm e}^{-\frac{z^3}{3}} \, 2^\frac13 \, \Ai\le(\omega\xi\ri),
\eea
where $\xi=2^{-\frac{2}{3}}\le(z^2-2\,t_0\ri),~\omega= {\rm e}^{\frac{2\pi \sqrt{-1}}{3}}.$ Noting that
\be
\xi_x=-2^{\frac13}, \quad \xi_z=2^{\frac13}z
\ee
we have
\bea
&&\psi_x(z;t_0,0,0,\dots)=-\sqrt{2\pi z}\,  {\rm e}^\frac{z^3}{3}2^\frac23 \, \Ai'(\xi),\\
&&\psi_x^*(z;t_0,0,0,\dots)=- {\rm e}^{\frac56\pi \sqrt{-1}}\sqrt{2\pi z}\,  {\rm e}^{-\frac{z^3}{3}}2^\frac23 \, \Ai'(\omega \xi),\\
&&\psi_z(z;t_0,0,0,\dots)=\sqrt{2\pi z}\,  {\rm e}^\frac{z^3}{3}2^\frac13 \le[ \le(\frac 1{2z} +z^2\ri) \Ai(\xi)+2^{\frac13}z\,\Ai'(\xi)\ri],\\
&&\psi_z^*(z;t_0,0,0,\dots)=-\sqrt{2\pi z}\,  {\rm e}^\frac{-z^3}{3}2^\frac13 \le[ {\rm e}^{\frac16\pi i} \le(\frac {1}{-2z} +z^2\ri) \Ai(\omega\xi)- {\rm e}^{\frac56\pi i}2^\frac13 z\, \Ai'(\omega \xi)\ri],\\
&&\psi_{zx}(z;t_0,0,0,\dots)=-\sqrt{2\pi z}\,  {\rm e}^\frac{z^3}{3}2^\frac23 \,\le[ \le(\frac 1{2z} +z^2\ri) \Ai'(\xi)+2^{-\frac13}z\le(z^2-2 t_0\ri)\Ai(\xi)\ri],\\
&&\psi^*_{zx}(z;t_0,0,0,\dots)=\sqrt{2\pi z}\,  {\rm e}^\frac{-z^3}{3}2^\frac23  \,\le[  {\rm e}^{\frac56 \pi i}\le(\frac 1{-2z} +z^2\ri)  \Ai'(\omega\xi)- {\rm e}^{\frac16 \pi i}2^{-\frac13}z\le(z^2-2t_0 \ri) \Ai(\omega\xi)\ri].
\eea

\section{Generalized Kac--Schwarz operator and higher Weil--Petersson volumes}\label{A-B}

Let us first compute $A(z;\bds)=\psi^\kappa(z; {\bf 0}; {\bf s}),\,B(z;\bds)=\psi^\kappa_x(z; {\bf 0}; {\bf s})$ associated with higher Weil--Petersson volumes by using the technique of generalized Kac--Schwarz operators explained above.

Theorem \ref{KMZ-thm} together with eq. \eqref{string} implies that  $Z^{\kappa}(\bdt;\bds)$ satisfies the string equation 
\be\label{string-wp}
\sum_{k\geq 0} \tilde{t}_{k+1} \frac{\p Z^{\kappa}}{\p t_k}+\frac{\tilde t_0^{\,2}}{2}  Z^{\kappa}=0,
\ee
where $\tilde{t}_0=t_0,\, \tilde t_k=t_k-h_{k-1}(-\bds),\,k\geq 1$, namely,
\be
c_0=0,\quad c_{k}=h_{k-1}(-\bds),\,k\geq 1.
\ee

The corresponding generalized Kac--Schwarz operator (see Definition \ref{def-KS}) reads as follows
\be
S_z=\frac{1}{z}\p_z-\frac{1}{2z^2} - z - \sum_{k=1}^\infty \frac{h_{k}(-\bds)}{(2k+1)!!} z^{2k+1}=S_z^{^{WK}}\!\! -\sum_{k=1}^\infty \frac{h_{k}(-\bds)}{(2k+1)!!} z^{2k+1}.
\ee
As before let $f(z;x;\bds)=\psi^\kappa(z;x,\bdzero;\bds),\,u_0(x;\bds)=\p_x^2 \log Z^{\kappa}(x,\bdzero;\bds).$ Then according to Prop. \ref{cor-string} we have
\bea
\label{S1-wp} && S_z f(z;x;\bds)=-f_x(z;x;\bds) - \sum_{k\geq 1} h_k(-\bds) \psi^\kappa_{t_k}(z; x,\bdzero;\bds),\\
\label{S1-der-wp} && S_z f_x(z;x;\bds)= -(z^2-2\, u_0(x;\bds)) f(z;x;\bds) - \sum_{k\geq 1} h_k(-\bds)  \psi^\kappa_{t_0 t_k}(z; x,\bdzero;\bds),\\
\label{S3-wp} && \psi^\kappa_{t_k}(z; x,\bdzero;\bds)=\frac{1}{(2k+1)!!} \sum_{i=0}^k (2i-1)!!\, z^{2k-2i}\le(\Omega_{i-1;0}|_{u\rightarrow u_0} f_x(z;x;\bds) - \frac12 f(z;x;\bds)\p_x\Omega_{i-1;0}|_{u\rightarrow u_0} \ri),\\
 && \psi^\kappa_{t_0 t_k}(z; x,\bdzero)=\frac{1}{(2k+1)!!} \sum_{i=0}^k (2i-1)!!\, z^{2k-2i}\le(\p_x\Omega_{i-1;0}|_{u\rightarrow u_0} f_x(z;x)+\Omega_{i-1;0}|_{u\rightarrow u_0} \le(z^2-2u_0(x)\ri) f(z;x)\ri.\nn\\
&&\le.\qquad\qquad - \frac12 f_x(z;x)\p_x\Omega_{i-1;0}|_{u\rightarrow u_0}- \frac12 f(z;x)\p_x^2\Omega_{i-1;0}|_{u\rightarrow u_0} \ri)\label{S4-wp}
\eea
Taking $x=0$ in the above equations \eqref{S1-wp},\eqref{S1-der-wp} we arrive at
\bl
The functions $A(z;\bds)$ and $B(z;\bds)$ satisfy
\bea
\label{S1-wp-0} && S_z A(z;\bds)=-B(z;\bds) - \sum_{k\geq 1} h_k(-\bds) \, \psi^\kappa_{t_k}(z; \bdzero;\bds),\\
\label{S1-der-wp-0} && S_z B(z;\bds)= -(z^2-2\, u_0(0;\bds)) A(z;\bds) - \sum_{k\geq 1} h_k(-\bds) \, \psi^\kappa_{t_0 t_k}(z;\bdzero;\bds).
\eea
\el 
To solve the above equations \eqref{S1-wp-0}, \eqref{S1-der-wp-0} we can expand them as formal series in $\bds$ and compare the coefficients. At the linear approximation in $\bds$, employing Lemma \ref{lab} and comparing the coefficients of $s_j,\,j\geq 1,$ we get
\bea
&& \!\!\!\!\!\!\!\!\! \frac{z^{2j+1}}{(2j+1)!!} A^{(0)}(z) + S_z^{^{WK}}\!\! A^{(j)}(z) =- B^{(j)}(z)+ \frac{z^{2j}}{(2j+1)!!}B^{(0)}(z)\nn\\
&& \qquad\qquad\qquad\qquad\qquad +\sum_{i=1}^j \frac{(2i-1)!!}{(2j+1)!!} z^{2j-2i} \le(\langle\tau_{i-1}\tau_0\rangle\, B^{(0)}(z) - \frac12 \langle\tau_{i-1}\tau_0^2\rangle A^{(0)}(z)\ri),\label{determ-ab1}\\
&& \!\!\!\!\!\!\!\!\! \frac{z^{2j+1}}{(2j+1)!!} B^{(0)}(z)+ S_z^{^{WK}}\!\! B^{(j)}(z)=-z^2 A^{(j)}(z)+2\langle\tau_0^2\kappa_j\rangle A^{(j)}(z) + \frac{z^{2j+2}}{(2j+1)!!} A^{(0)}(z)\nn\\
&& +\sum_{i=1}^j \frac{(2i-1)!!}{(2j+1)!!}z^{2j-2i} \le(\frac12 \langle\tau_0^2\tau_{i-1}\rangle B^{(0)}(z)+\langle\tau_0\tau_{i-1}\rangle z^2 A^{(0)}(z) - \frac12 A^{(0)}(z)\langle\tau_0^3\tau_{i-1}\rangle \ri). \label{determ-ab2} 
\eea
Also, taking derivatives w.r.t. $x$ in the string equation \eqref{string-wp} and taking $\bdt=0$ yields
\bl \label{string-0}
For any $m\geq 0$, we have
\be\label{string-wp-0}
\sum_{k\geq 0} h_k(-\bds) \frac{\p^{m+1} \log Z^{\kappa}}{\p t_0^m \p t_k}(\bdzero;\bds)=0.
\ee
\el
Comparing the coefficients of $s_j$ in \eqref{string-wp-0} we find
\be
-\langle\tau_0^{m} \tau_j \rangle+\langle\tau_0^{m+1}\kappa_j \rangle=0.
\ee
Substituting this expression into \eqref{determ-ab2} and using \eqref{no-one},\eqref{string} we can solve out $A^{(j)}(z),B^{(j)}(z)$, e.g. in the case $j=1$ we have
\bl $A^{(1)}(z)$ and $B^{(1)}(z)$ satisfy the following ODE system
\bea
\label{(1)-ode1} S_z^{^{WK}}\!\! A^{(1)}(z) + B^{(1)}(z)&=&-\frac{1}6 A^{(0)}(z) - \frac{z^3} 3 A^{(0)}(z)+\frac{z^2}{3}B^{(0)}(z),\\
\label{(1)-ode2} S_z^{^{WK}}\!\! B^{(1)}(z) + z^2 A^{(1)}(z)&=&\frac{1}6 B^{(0)}(z) - \frac{z^3}3 B^{(0)}(z) + \frac{z^4}{3} A^{(0)}(z).
\eea
Moreover, the solution of this ODE system is unique under the boundary condition
\be
A^{(1)}(z)=\mathcal{O}(z^{-1}),\qquad z\rightarrow \infty.
\ee 
\el
Recall that $A^{(0)}(z)=c(z),\,B^{(0)}(z)=z \, q(z);$ we obtain the explicit solution of the ODE system \eqref{(1)-ode1},\eqref{(1)-ode2}:
\bea
&& A^{(1)}(z)=\frac{1}{5}\sum_{g=1}^\infty g\,C_g \, z^{-3g+2},\\
&& B^{(1)}(z)=\frac1{120}\sum_{g=0}^\infty  \le(36g^2+48g-5\ri) C_g\, z^{-3g}.
\eea
Similarly, we know from the system \eqref{S1-wp-0}, \eqref{S1-der-wp-0} that for any partition $\lambda,$
\bea
&& S_z^{^{WK}}\!\! A^{\lambda}(z) + B^{\lambda}(z) = \hbox{combinations of } A^{\mu}(z), B^{\mu}(z)\,\hbox{ with } |\mu|<|\lambda|,\\
&& S_z^{^{WK}}\!\! B^{\lambda}(z) +z^2 A^{\lambda}(z) = \hbox{combinations of } A^{\mu}(z), B^{\mu}(z)\,\hbox{ with } |\mu|<|\lambda|.
\eea
It is not difficult to see that the combination coefficients in the above formul\ae\ can be obtained in a recursive procedure by applying Lemma \ref{string-0} and the {\it  Main Theorem}. Details are omitted. So, in principle one can obtain all the initial data $A^{\lambda},B^\lambda,\, |\lambda|>0$ from \eqref{S1-wp-0}, \eqref{S1-der-wp-0}, Lemma~\ref{string-0} with the knowledge of $A^{(0)},B^{(0)}$.

Now let
$
F^{\kappa}_n(z_1,\dots,z_n;\bds):=\nabla (z_1) \cdots \nabla (z_n)  \log Z^{\kappa}(\bdt;\bds)\,|_{\bdt=0}
$
be the $n$-point generating function corresponding to the partition function $Z^{\kappa}$.
We have for $n=1,$
\bea
F^{\kappa}_1(z;\bds)&=&\frac{1}{4\,z}\le(-A(z;\bds)\, B_z(-z;\bds)+B_z(z;\bds)\, A(-z;\bds)+B(z;\bds)\, A_z(-z;\bds)-A_z(z;\bds)\, B(-z;\bds)\ri)\nn\\
&=& F_1(z)+\frac{\fs_1}{4\,z}\Big(-A^{^{WK}}\!\!(z) \, B^{(1)}_z(-z)-A^{(1)}(z)\,B_{z}^{^{WK}}\!\!(-z)+ B_{z}^{^{WK}}\!\!(z)\, A^{(1)}(-z)+B^{(1)}_z(z)\,A^{^{WK}}\!\!(-z)\nn\\
&& +\, B^{^{WK}}\!\!(z)\, A^{(1)}_z(-z)+B^{(1)}(z)\,A_{z}^{^{WK}}\!\!(-z)-A_{z}^{^{WK}}\!\!(z)\, B^{(1)}(-z)-A^{(1)}_z(z)B^{^{WK}}\!\!(-z)\Big)\nn\\
&& +\, \hbox{higher order terms}\nn\\
&=& F_1(z)+ \fs_1 \sum_{g=1}^\infty \frac{(12g^2-12g+5)(6g-5)!!}{5 \cdot 24^g\cdot g!}z^{-6g+4}+\, \hbox{higher order terms}.
\eea
The last equality uses similar derivation as in Lemma \ref{2F3}. We read off from the above expression that
\be\label{one-point-kappa1}
\langle\kappa_1 \tau_{3g-3}\rangle_{g,1}=\frac{12g^2-12g+5}{5 \cdot 24^g\cdot g!},\qquad g\geq 1.
\ee

For $n\geq 2,$ let $M^{\kappa}(z;\bds)=\Theta(z;\bdzero;\bds)$ then we have
\be
F^{\kappa}_n(z_1,\dots,z_n;\bds)&=&-\frac{1}{n}\sum_{r\in S_{n}}  \frac{{\rm Tr} \le( M^\kappa (z_{r_1};\bds)\cdots M^\kappa (z_{r_{n}};\bds)\ri)}{
\prod_{j=1}^n(z_{r_j}^2-z_{r_{j+1}}^2 )} -\delta_{n,2}\frac{z_1^2+z_2^2}{(z_1^2-z_2^2)^2},
\ee
Here the matrix-value function $M^{\kappa}(z;\bds)$ has the form
\be
M^{\kappa}(z;\bds)=M(z)+\fs_1
\left(
\begin{array}{cc}
M^{(1)}_{11}(z) & M^{(1)}_{12}(z)\\
M^{(1)}_{21}(z) & M^{(1)}_{22}(z)\\
\end{array}
\right)+\hbox{higher order terms}
\ee
with
\bea
M^{(1)}_{11}(z)&=&-\frac{1}{2}\le(B^{^{WK}}\!\!(z)\,A^{(1)}(-z)+B^{(1)}(z)\,A^{^{WK}}\!\!(-z)+A^{^{WK}}\!\!(z)\,B^{(1)}(-z)+A^{(1)}(z)\,B^{^{WK}}\!\!(-z)\ri),\\
M^{(1)}_{12}(z)&=&-A^{^{WK}}\!\!(z)\,A^{(1)}(-z)-A^{(1)}(z)\,A^{^{WK}}\!\!(-z),\\
M^{(1)}_{21}(z)&=&B^{^{WK}}\!\!(z)\,B^{(1)}(-z)+B^{(1)}(z)\,B^{^{WK}}\!\!(-z),\\
M^{(1)}_{22}(z)&=&-M^{(1)}_{12}(z).
\eea

\begin{landscape}
\section{Tables of some intersection numbers}\label{table}
The nonzero two point intersection numbers $\langle \tau_k \tau_\ell\rangle$ with $2\leq k \leq \ell\leq 30$:
{\tiny
\bea\nn
\begin {array}{ccc} \langle \tau_{{3}}\tau_{{2}}\rangle ={\frac {29}{5760}}&
\langle \tau_{{2}}\tau_{{6}}\rangle ={\frac {77}{414720}}&\langle \tau_{{2}}\tau_{{9}}\rangle ={
\frac {149}{39813120}}\\ \noalign{\medskip}\langle \tau_{{2}}\tau_{{12}}\rangle ={
\frac {49}{955514880}}&\langle \tau_{{2}}\tau_{{15}}\rangle ={\frac {73}{
137594142720}}&\langle \tau_{{2}}\tau_{{18}}\rangle ={\frac {509}{115579079884800}}
\\ \noalign{\medskip}\langle \tau_{{2}}\tau_{{21}}\rangle ={\frac {677}{
22191183337881600}}&\langle \tau_{{2}}\tau_{{24}}\rangle ={\frac {869}{
4793295600982425600}}&\langle \tau_{{2}}\tau_{{27}}\rangle ={\frac {31}{
32868312692450918400}}\\ \noalign{\medskip}\langle \tau_{{2}}\tau_{{30}}\rangle ={
\frac {53}{12148128371129859440640}}&\langle \tau_{{3}}\tau_{{5}}\rangle ={\frac {
503}{1451520}}&\langle \tau_{{3}}\tau_{{8}}\rangle ={\frac {947}{92897280}}
\\ \noalign{\medskip}\langle \tau_{{3}}\tau_{{11}}\rangle ={\frac {307}{1672151040}}
&\langle \tau_{{3}}\tau_{{14}}\rangle ={\frac {2267}{963158999040}}&\langle \tau_{{3}}\tau_
{{17}}\rangle ={\frac {449}{19263179980800}}\\ \noalign{\medskip}\langle \tau_{{3}}
\tau_{{20}}\rangle ={\frac {4163}{22191183337881600}}&\langle \tau_{{3}}\tau_{{23}}\rangle 
={\frac {761}{599161950122803200}}&\langle \tau_{{3}}\tau_{{26}}\rangle ={\frac {
1327}{178949702436677222400}}\\ \noalign{\medskip}\langle \tau_{{3}}\tau_{{29
}}\rangle ={\frac {8087}{212592246494772540211200}}&\langle {\tau_{{4}}}^{2}\rangle ={
\frac {607}{1451520}}&\langle \tau_{{4}}\tau_{{7}}\rangle ={\frac {1781}{92897280}}
\\ \noalign{\medskip}\langle \tau_{{4}}\tau_{{10}}\rangle ={\frac {29}{61931520}}&\langle 
\tau_{{4}}\tau_{{13}}\rangle ={\frac {7297}{963158999040}}&\langle \tau_{{4}}\tau_{{
16}}\rangle ={\frac {349}{3852635996160}}\\
\end{array}
\eea
}
\end{landscape}
\begin{landscape}
{\tiny
\bea\nn
\begin{array}{ccc}
\noalign{\medskip}\langle \tau_{{12
}}\tau_{{20}}\rangle ={\frac {7702634988677}{6843840463241882538659020800}}&\langle 
\tau_{{12}}\tau_{{23}}\rangle ={\frac {532195893421}{
37643736696212035353968640000}}&\langle \tau_{{12}}\tau_{{26}}\rangle ={\frac {
545177136989}{3772298666820406069155594240000}}\\ \noalign{\medskip}\langle 
\tau_{{12}}\tau_{{29}}\rangle ={\frac {32414448521083}{
26199484771529294090147714826240000}}&\langle {\tau_{{13}}}^{2}\rangle ={\frac {
2701472413}{774117239558661734400}}&\langle \tau_{{13}}\tau_{{16}}\rangle ={\frac {
4131287269}{45414878054108155084800}}\\ \noalign{\medskip}\langle \tau_{{13}}
\tau_{{19}}\rangle ={\frac {3898864039667}{2281280154413960846219673600}}&\langle 
\tau_{{13}}\tau_{{22}}\rangle ={\frac {926415021131}{
37643736696212035353968640000}}&\langle \tau_{{13}}\tau_{{25}}\rangle ={\frac {
1070162502821}{3772298666820406069155594240000}}\\ \noalign{\medskip}\langle 
\tau_{{13}}\tau_{{28}}\rangle ={\frac {874463851081}{
323450429278139433211700183040000}}&\langle \tau_{{14}}\tau_{{15}}\rangle ={\frac {
427374701}{4128625277646195916800}}&\langle \tau_{{14}}\tau_{{18}}\rangle ={\frac {
1209992413697}{526449266403221733743001600}}\\ \noalign{\medskip}\langle \tau
_{{14}}\tau_{{21}}\rangle ={\frac {41072588543}{1075535334177486724399104000}
}&\langle \tau_{{14}}\tau_{{24}}\rangle ={\frac {110706467197}{
221899921577670945244446720000}}&\langle \tau_{{14}}\tau_{{27}}\rangle ={\frac {
212494716887933}{39988687282860501506014933155840000}}
\\ \noalign{\medskip}\langle \tau_{{14}}\tau_{{30}}\rangle ={\frac {6739031639609}{
142089673254423807948645268979712000}}&\langle \tau_{{15}}\tau_{{17}}\rangle ={
\frac {48139483517}{17548308880107391124766720}}&\langle \tau_{{15}}\tau_{{20
}}\rangle ={\frac {398801587363}{7528747339242407070793728000}}
\\
\end {array}
\eea}
\end{landscape}
\begin{landscape}
The nonzero three point intersection numbers $\langle \tau_{j_1} \tau_{j_2} \tau_{j_3}\rangle$ with $2\leq j_1\leq j_2\leq j_3\leq 22$:
{\tiny 
\bea\nn
 \begin {array}{cccc} \langle {\tau_{{2}}}^{3}\rangle ={\frac {7}{240}}&\langle {
\tau_{{2}}}^{2}\tau_{{5}}\rangle ={\frac {17}{5760}}&\langle {\tau_{{2}}}^{2}\tau_{{
8}}\rangle ={\frac {53}{460800}}&\langle {\tau_{{2}}}^{2}\tau_{{11}}\rangle ={\frac {257}{
99532800}}\\ \noalign{\medskip}\langle {\tau_{{2}}}^{2}\tau_{{14}}\rangle ={\frac {
379}{9555148800}}&\langle {\tau_{{2}}}^{2}\tau_{{17}}\rangle ={\frac {7}{15288238080
}}&\langle {\tau_{{2}}}^{2}\tau_{{20}}\rangle ={\frac {139}{33022594252800}}&\langle \tau_{
{2}}\tau_{{3}}\tau_{{4}}\rangle ={\frac {1121}{241920}}\\ \noalign{\medskip}\langle 
\tau_{{2}}\tau_{{3}}\tau_{{7}}\rangle ={\frac {16243}{58060800}}&\langle \tau_{{2}}
\tau_{{3}}\tau_{{10}}\rangle ={\frac {35521}{4180377600}}&\langle \tau_{{2}}\tau_{{3
}}\tau_{{13}}\rangle ={\frac {9433}{57330892800}}&\langle \tau_{{2}}\tau_{{3}}\tau_{
{16}}\rangle ={\frac {22073}{9631589990400}}\\ \noalign{\medskip}\langle \tau_{{2}}
\tau_{{3}}\tau_{{19}}\rangle ={\frac {4889}{198135565516800}}&\langle \tau_{{2}}\tau
_{{3}}\tau_{{22}}\rangle ={\frac {35839}{166433875034112000}}&\langle \tau_{{2}}\tau
_{{4}}\tau_{{6}}\rangle ={\frac {1081}{2322432}}&\langle \tau_{{2}}\tau_{{4}}\tau_{{
9}}\rangle ={\frac {11827}{597196800}}\\ \noalign{\medskip}\langle \tau_{{2}}\tau_{{
4}}\tau_{{12}}\rangle ={\frac {28279}{57330892800}}&\langle \tau_{{2}}\tau_{{4}}\tau
_{{15}}\rangle ={\frac {80897}{9631589990400}}&\langle \tau_{{2}}\tau_{{4}}\tau_{{18
}}\rangle ={\frac {21179}{198135565516800}}&\langle \tau_{{2}}\tau_{{4}}\tau_{{21}}\rangle 
={\frac {35831}{33286775006822400}}\\ \noalign{\medskip}\langle \tau_{{2}}{
\tau_{{5}}}^{2}\rangle ={\frac {533}{967680}}&\langle \tau_{{2}}\tau_{{5}}\tau_{{8}}
\rangle ={\frac {12491}{364953600}}&\langle \tau_{{2}}\tau_{{5}}\tau_{{11}}\rangle ={\frac 
{5893}{5255331840}}&\langle \tau_{{2}}\tau_{{5}}\tau_{{14}}\rangle ={\frac {16721}{
706316599296}}\\
\end{array}
\eea
\bea\nn
\begin{array}{cccc}
\\ \noalign{\medskip}\langle \tau_{{3}}\tau_{{9}}\tau_{{21}}\rangle ={\frac {
2961244699003}{14877914050525831605780480}}&\langle \tau_{{3}}\tau_{{10}}\tau
_{{11}}\rangle ={\frac {2922785743}{107280501949071360}}&\langle \tau_{{3}}\tau_{{10
}}\tau_{{14}}\rangle ={\frac {11090288141}{11510566797359185920}}&\langle \tau_{{3}}
\tau_{{10}}\tau_{{17}}\rangle ={\frac {120835139197}{5228686618071662592000}}
\\ \noalign{\medskip}\langle \tau_{{3}}\tau_{{10}}\tau_{{20}}\rangle ={\frac {
15158746939103}{37194785126314579014451200}}&\langle \tau_{{3}}\tau_{{11}}
\tau_{{13}}\rangle ={\frac {13983412199}{11510566797359185920}}&\langle \tau_{{3}}
\tau_{{11}}\tau_{{16}}\rangle ={\frac {210148115}{5975641849224757248}}&\langle \tau
_{{3}}\tau_{{11}}\tau_{{19}}\rangle ={\frac {621508694340991}{
855480057905235317332377600}}\\ \noalign{\medskip}\langle \tau_{{3}}\tau_{{11
}}\tau_{{22}}\rangle ={\frac {272315357551}{23868311311004007923712000}}&\langle 
\tau_{{3}}{\tau_{{12}}}^{2}\rangle ={\frac {3020419917227}{
2302113359471837184000}}&\langle \tau_{{3}}\tau_{{12}}\tau_{{15}}\rangle ={\frac {
25217780039}{543240168111341568000}}&\langle \tau_{{3}}\tau_{{12}}\tau_{{18}}
\rangle ={\frac {2983241978467727}{2632246332016108668715008000}}
\\ \noalign{\medskip}\langle \tau_{{3}}\tau_{{12}}\tau_{{21}}\rangle ={\frac {
19606706386613}{954732452440160316948480000}}&\langle \tau_{{3}}\tau_{{13}}
\tau_{{14}}\rangle ={\frac {14476869931}{271620084055670784000}}&\langle \tau_{{3}}
\tau_{{13}}\tau_{{17}}\rangle ={\frac {340678837517771}{
219353861001342389059584000}}&\langle \tau_{{3}}\tau_{{13}}\tau_{{20}}\rangle ={
\frac {2230392447653}{68195175174297165496320000}}
\\
\end{array}
\eea
\bea\nn
 \begin {array}{ccc} \langle \tau_{{5}}\tau_{{9}}\tau_{{19}}\rangle ={\frac {
207147188280589}{58448948055637195594137600}}&\langle \tau_{{5}}\tau_{{9}}
\tau_{{22}}\rangle ={\frac {390376552056761}{7013873766676463471296512000}}&\langle 
\tau_{{5}}\tau_{{10}}\tau_{{12}}\rangle ={\frac {4823355455459}{
633081173854755225600}}\\ \noalign{\medskip}\langle \tau_{{5}}\tau_{{10}}\tau
_{{15}}\rangle ={\frac {1025122387}{3802681176779390976}}&\langle \tau_{{5}}\tau_{{
10}}\tau_{{18}}\rangle ={\frac {41429423338949}{6294502098299390294753280}}&\langle 
\tau_{{5}}\tau_{{10}}\tau_{{21}}\rangle ={\frac {1171129439636251}{
9819423273347048859815116800}}\\ \noalign{\medskip}\langle \tau_{{5}}{\tau_{{
11}}}^{2}\rangle ={\frac {8737971317}{1055135289757925376}}&\langle \tau_{{5}}\tau_{
{11}}\tau_{{14}}\rangle ={\frac {12665466509}{34857910787144417280}}&\langle \tau_{{
5}}\tau_{{11}}\tau_{{17}}\rangle ={\frac {212901240024101}{
20107437258456385663795200}}\\ \noalign{\medskip}\langle \tau_{{5}}\tau_{{11}
}\tau_{{20}}\rangle ={\frac {364917172905469}{1636570545557841476635852800}}&
\langle \tau_{{5}}\tau_{{12}}\tau_{{13}}\rangle ={\frac {8815163665079}{
20914746472286650368000}}&\langle \tau_{{5}}\tau_{{12}}\tau_{{16}}\rangle ={\frac {
1021925948176829}{68939784886136179418726400}}\\ \noalign{\medskip}\langle 
\tau_{{5}}\tau_{{12}}\tau_{{19}}\rangle ={\frac {359078501200526983}{
981942327334704885981511680000}}&\langle \tau_{{5}}\tau_{{12}}\tau_{{22}}\rangle ={
\frac {3918464405395483}{579713545121665344451117056000}}&\langle \tau_{{5}}
\tau_{{13}}\tau_{{15}}\rangle ={\frac {113547322348489}{
6267253171466925401702400}}\\ \noalign{\medskip}\langle \tau_{{5}}\tau_{{13}}
\tau_{{18}}\rangle ={\frac {39897610190525287}{75534025179592683537039360000}
}&\langle \tau_{{5}}\tau_{{13}}\tau_{{21}}\rangle ={\frac {1306154784535829}{
115942709024333068890223411200}}&\langle \tau_{{5}}{\tau_{{14}}}^{2}\rangle ={\frac 
{20229696415813}{1044542195244487566950400}}\\
\end{array}
\eea
\bea\nn
\begin{array}{ccc}
\\ \noalign{\medskip}\langle \tau_{{7}}\tau_{{10}}\tau_{{13}}\rangle ={\frac {
9938691090457}{7996814827639013376000}}&\langle \tau_{{7}}\tau_{{10}}\tau_{{
16}}\rangle ={\frac {4980158843027}{113935571258552038195200}}&\langle \tau_{{7}}
\tau_{{10}}\tau_{{19}}\rangle ={\frac {1376470684118437933}{
1276525025535116351775965184000}}\\ \noalign{\medskip}\langle \tau_{{7}}\tau_
{{10}}\tau_{{22}}\rangle ={\frac {11102336601883907}{
557029102051687135320421171200}}&\langle \tau_{{7}}\tau_{{11}}\tau_{{12}}\rangle ={
\frac {11667166895743}{7996814827639013376000}}&\langle \tau_{{7}}\tau_{{11}}
\tau_{{15}}\rangle ={\frac {134464317788107}{2144060295501842900582400}}
\\ \noalign{\medskip}\langle \tau_{{7}}\tau_{{11}}\tau_{{18}}\rangle ={\frac {
179539670079186899}{98194232733470488598151168000}}&\langle \tau_{{7}}\tau_{{
11}}\tau_{{21}}\rangle ={\frac {4344392757300223}{
111405820410337427064084234240}}&\langle \tau_{{7}}\tau_{{12}}\tau_{{14}}\rangle ={
\frac {476387883510529}{6125886558576694001664000}}
\\ \noalign{\medskip}\langle \tau_{{7}}\tau_{{12}}\tau_{{17}}\rangle ={\frac {
885729055722730831}{327314109111568295327170560000}}&\langle \tau_{{7}}\tau_{
{12}}\tau_{{20}}\rangle ={\frac {5977884495667313}{
89124656328269941651267387392}}&\langle \tau_{{7}}{\tau_{{13}}}^{2}\rangle ={\frac {
108537290563007}{1299430482122329030656000}}\\ \noalign{\medskip}\langle \tau
_{{7}}\tau_{{13}}\tau_{{16}}\rangle ={\frac {32804778301136311}{
9351831688901951295062016000}}&\langle \tau_{{7}}\tau_{{13}}\tau_{{19}}\rangle ={
\frac {20630745781051999}{202556037109704412843789516800}}&\langle \tau_{{7}}
\tau_{{13}}\tau_{{22}}\rangle ={\frac {146630065607595893}{
66989120769614662025462415360000}}\\
\end{array}
\eea
\bea\nn
\begin{array}{ccc}
\\ \noalign{\medskip}\langle {\tau
_{{10}}}^{3}\rangle ={\frac {682434427273}{239904444829170401280}}&\langle {\tau_{{
10}}}^{2}\tau_{{13}}\rangle ={\frac {293531282249}{1535388446906690568192}}&\langle 
{\tau_{{10}}}^{2}\tau_{{16}}\rangle ={\frac {4149376522392359}{
516811751228792045253427200}}\\ \noalign{\medskip}\langle {\tau_{{10}}}^{2}
\tau_{{19}}\rangle ={\frac {2726948053674176819}{
11697611143085429841728844595200}}&\langle {\tau_{{10}}}^{2}\tau_{{22}}\rangle ={
\frac {156276817283671919}{31193629714894479577943585587200}}&\langle \tau_{{
10}}\tau_{{11}}\tau_{{12}}\rangle ={\frac {594401002226287}{
2648545070914041230131200}}\\ \noalign{\medskip}\langle \tau_{{10}}\tau_{{11}
}\tau_{{15}}\rangle ={\frac {541223073795413}{46982886475344731386675200}}&\langle 
\tau_{{10}}\tau_{{11}}\tau_{{18}}\rangle ={\frac {355688890808287217}{
899816241775802295517603430400}}&\langle \tau_{{10}}\tau_{{11}}\tau_{{21}}\rangle ={
\frac {550366193566912447}{56148533486810063240298454056960}}
\\ \noalign{\medskip}\langle \tau_{{10}}\tau_{{12}}\tau_{{14}}\rangle ={\frac {
67111662223673413}{4698288647534473138667520000}}&\langle \tau_{{10}}\tau_{{
12}}\tau_{{17}}\rangle ={\frac {1754731878725777269}{
2999387472586007651725344768000}}&\langle \tau_{{10}}\tau_{{12}}\tau_{{20}}\rangle =
{\frac {94662985826333358373}{5614853348681006324029845405696000}}
\\ \noalign{\medskip}\langle \tau_{{10}}{\tau_{{13}}}^{2}\rangle ={\frac {
72082892376709001}{4698288647534473138667520000}}&\langle \tau_{{10}}\tau_{{
13}}\tau_{{16}}\rangle ={\frac {324950341754567773}{
428483924655143950246477824000}}&\langle \tau_{{10}}\tau_{{13}}\tau_{{19}}\rangle ={
\frac {1774660442594675773}{69319177144209954617652412416000}}
\\
\end{array}
\eea
\bea\nn
\begin{array}{ccc}
\\ \noalign{\medskip}\langle {\tau_{{14}}}^{2}\tau_{{17}}\rangle ={\frac {
23414499614949743}{4218918642758340433196089344000}}&\langle {\tau_{{14}}}^{2
}\tau_{{20}}\rangle ={\frac {462767707930924427}{
2206272401812993186538754932736000}}&\langle \tau_{{14}}\tau_{{15}}\tau_{{16}
}\rangle ={\frac {2107304965306301063}{336308086094164851674773979136000}}
\\ \noalign{\medskip}\langle \tau_{{14}}\tau_{{15}}\tau_{{19}}\rangle ={\frac {
3672285682784833909}{13237634410877959119232529596416000}}&\langle \tau_{{14}
}\tau_{{15}}\tau_{{22}}\rangle ={\frac {121296410843128764497}{
13719002934909884905386439763558400000}}&\langle \tau_{{14}}\tau_{{16}}\tau_{
{18}}\rangle ={\frac {18511317621367751831}{
56462563099459050120808136441856000}}\\ \noalign{\medskip}\langle \tau_{{14}}
\tau_{{16}}\tau_{{21}}\rangle ={\frac {4730560022267075714303}{
392363483938422708294052177237770240000}}&\langle \tau_{{14}}{\tau_{{17}}}^{2
}\rangle ={\frac {932250523657603974797}{
2689813769877007526588498722160640000}}&\langle \tau_{{14}}\tau_{{17}}\tau_{{
20}}\rangle ={\frac {2607872832951318082813}{
176060537664676856285792643632332800000}}\\ \noalign{\medskip}\langle \tau_{{
14}}\tau_{{18}}\tau_{{19}}\rangle ={\frac {1352428754551622110128037}{
82396331627068768741750957219931750400000}}&\langle \tau_{{14}}\tau_{{18}}
\tau_{{22}}\rangle ={\frac {42976490465567970100061137}{
73167942484837066642674850011299394355200000}}&\langle \tau_{{14}}\tau_{{19}}
\tau_{{21}}\rangle ={\frac {128929471395477419779072391}{
190236650460576373270954610029378425323520000}}\\ \noalign{\medskip}\langle 
\tau_{{14}}{\tau_{{20}}}^{2}\rangle ={\frac {3219493188176607553369649}{
4529444058585151744546538334032819650560000}}&\langle \tau_{{14}}\tau_{{21}}
\tau_{{22}}\rangle ={\frac {69165220224212231630868593}{
2556075457925891038544100551746088494694400000}}&\langle {\tau_{{15}}}^{3}\rangle ={
\frac {407865477597219179}{61146924744393609395413450752000}}
\\
\end {array}
\eea}
\end{landscape}

The nonzero four point intersection numbers $\langle \tau_{j_1} \tau_{j_2} \tau_{j_3}\tau_{j_4}\rangle$ with $2\leq j_1\leq j_2\leq j_3\leq j_4\leq 9:$
{\tiny
\bea
\nn
 \begin {array}{cccc} \langle {\tau_{{2}}}^{3}\tau_{{4}}\rangle ={\frac {53}{
1152}}&\langle {\tau_{{2}}}^{3}\tau_{{7}}\rangle ={\frac {4907}{1382400}}&\langle {\tau_{{2
}}}^{2}{\tau_{{3}}}^{2}\rangle ={\frac {205}{3456}}&\langle {\tau_{{2}}}^{2}\tau_{{3
}}\tau_{{6}}\rangle ={\frac {7897}{1036800}}\\ \noalign{\medskip}\langle {\tau_{{2}}
}^{2}\tau_{{3}}\tau_{{9}}\rangle ={\frac {117871}{298598400}}&\langle {\tau_{{2}}}^{
2}\tau_{{4}}\tau_{{5}}\rangle ={\frac {7597}{691200}}&\langle {\tau_{{2}}}^{2}\tau_{
{4}}\tau_{{8}}\rangle ={\frac {27629}{33177600}}&\langle {\tau_{{2}}}^{2}\tau_{{5}}
\tau_{{7}}\rangle ={\frac {14237}{11059200}}\\ \noalign{\medskip}\langle {\tau_{{2}}
}^{2}{\tau_{{6}}}^{2}\rangle ={\frac {295693}{199065600}}&\langle {\tau_{{2}}}^{2}
\tau_{{6}}\tau_{{9}}\rangle ={\frac {155893}{1433272320}}&\langle {\tau_{{2}}}^{2}
\tau_{{7}}\tau_{{8}}\rangle ={\frac {365669}{2654208000}}&\langle {\tau_{{2}}}^{2}{
\tau_{{9}}}^{2}\rangle ={\frac {606960439}{56757583872000}}
\\ \noalign{\medskip}\langle \tau_{{2}}{\tau_{{3}}}^{2}\tau_{{5}}\rangle ={\frac {
1837}{129600}}&\langle \tau_{{2}}{\tau_{{3}}}^{2}\tau_{{8}}\rangle ={\frac {747709}{
696729600}}&\langle \tau_{{2}}\tau_{{3}}{\tau_{{4}}}^{2}\rangle ={\frac {177}{10240}
}&\langle \tau_{{2}}\tau_{{3}}\tau_{{4}}\tau_{{7}}\rangle ={\frac {2821883}{
1393459200}}\\ \noalign{\medskip}\langle \tau_{{2}}\tau_{{3}}\tau_{{5}}\tau_{
{6}}\rangle ={\frac {1924831}{696729600}}&\langle \tau_{{2}}\tau_{{3}}\tau_{{5}}\tau
_{{9}}\rangle ={\frac {5072357}{25082265600}}&\langle \tau_{{2}}\tau_{{3}}\tau_{{6}}
\tau_{{8}}\rangle ={\frac {4942457}{16721510400}}&\langle \tau_{{2}}\tau_{{3}}{\tau_
{{7}}}^{2}\rangle ={\frac {112020239}{334430208000}}\\
\end {array} 
\eea
}

~~~~
~~~~

\noindent Marco Bertola

\noindent Department of Mathematics and
Statistics, Concordia University, 1455 de Maisonneuve W., Montr\'eal, Qu\'ebec,  H3G 1M8,
Canada

\noindent SISSA, via Bonomea 265, Trieste 34136, Italy

\noindent Centre de recherches math\'ematiques, Universit\'e de Montr\'eal, C.~P.~6128, succ. centre ville, Montr\'eal,
Qu\'ebec, H3C 3J7, Canada

\noindent marco.bertola@concordia.ca, Marco.Bertola@sissa.it

~~~~~
~~~~~

\noindent Boris Dubrovin

\noindent SISSA, via Bonomea 265, Trieste 34136, Italy

\noindent dubrovin@sissa.it

~~~~~
~~~~~

\noindent Di Yang

\noindent SISSA, via Bonomea 265, Trieste 34136, Italy

\noindent dyang@sissa.it

\end{document}